\documentclass[letterpaper, 10 pt, conference]{ieeeconf}  % Comment this line out
\IEEEoverridecommandlockouts                              % This command is only
\usepackage{amsmath,amsfonts}
\usepackage[ruled]{algorithm2e}
\usepackage{subfigure}
\usepackage{graphicx}
\usepackage{cite}
\usepackage{color}
\usepackage{multirow}
\usepackage{booktabs} 
\usepackage{tabularx}
\usepackage{float}
\allowdisplaybreaks

\usepackage{hyperref} %add by myself
\usepackage{amsmath,amssymb}

\usepackage{amsthm}

\newtheorem{lemma}{Lemma}
\newtheorem{proposition}{Proposition}
\newtheorem{corollary}{Corollary}

\newtheorem{remark}{Remark}
\usepackage{graphicx}
\usepackage{subcaption}
\usepackage{xurl}

\usepackage{xcolor}

\title{\LARGE \bf Minimizing Bid Cost Recovery for Energy Storage with Uniform Pricing
}

\author{Yaxuan Yu, Jingguan Liu, and Cong Chen%
\thanks{\textit{Corresponding authors:~Cong Chen;~Jingguan Liu.}}
\thanks{Y. Yu and C. Chen are with the Thayer School of Engineering, 
Dartmouth College, Hanover, NH 03755, USA.
{\tt\small \{yaxuan.yu.th, cong.chen\}@dartmouth.edu}}
\thanks{J. Liu is with the School of Electrical and Electronic Engineering, 
Huazhong University of Science and Technology, Wuhan, China.
{\tt\small spencerplusmail@foxmail.com}.}
}

\begin{document}

\maketitle
\thispagestyle{empty}
\pagestyle{empty}

%%%%%%%%%%%%%%%%%%%%%%%%%%%%%%%%%%%%%%%%%%%%%%%%%%%%%%%%%%%%%%%%%%%%%%%%%%%%%%%%
\begin{abstract}
We study in-market uniform pricing and out-of-market bid cost recovery (BCR) payments in rolling-window dispatch for real-time power system operations with energy storage resources (ESRs). Due to intertemporal state-of-charge (SOC) constraints, ESR operations are temporally coupled, and existing in-market locational marginal pricing (LMP) may fail to compensate ESR's intertemporal opportunity costs, thereby triggering out-of-market BCR payments. We show that positive BCR is unavoidable when dispatched generators or ESRs have supply-side bids higher than the demand-side bid.
%We show that strictly positive BCR is unavoidable when the conditions
%for the existence of a zero-BCR uniform price fail.
%We characterize conditions under which a zero-BCR uniform price exists and show that, otherwise, strictly positive BCR is unavoidable. 
We further identify an intertemporal coupling indicator associated with binding SOC constraints and demonstrate empirically that positive BCR arises only when this indicator is active, revealing that BCR is fundamentally driven by intertemporal coupling. In the simulation, we compare a BCR-minimizing uniform pricing scheme (UP-BCR) with existing real-time pricing methods under forecast uncertainty and show that UP-BCR substantially reduces BCR and demand payments relative to LMP while maintaining zero merchandising surplus in a copper-plate model.

% We study uniform pricing in  multi-interval electricity markets with energy storage resources under rolling-window dispatch. Because storage operation is coupled over time through state-of-charge constraints, locational marginal pricing may fail to reflect intertemporal opportunity costs, causing market revenues to be insufficient to support dispatch-following incentives and leading to out-of-market bid cost recovery (BCR) payments. To address this issue, we develop uniform-pricing formulations that determine a single settlement price in each interval to reduce total bid cost recovery. We characterize when a zero-bid-cost-recovery uniform price exists and show that, otherwise, strictly positive BCR is unavoidable. We further identify a window-level indicator associated with binding intertemporal state-of-charge constraints and show empirically that positive BCR arises only when this indicator is active. Monte Carlo simulations under forecast uncertainty demonstrate that the proposed pricing schemes substantially reduce BCR and demand payments relative to locational marginal pricing, while maintaining zero merchandising surplus.
\end{abstract}

%%%%%%%%%%%%%%%%%%%%%%%%%%%%%%%%%%%%%%%%%%%%%%%%%%%%%%%%%%%%%%%%%%%%%%%%%%%%%%%%
\section{INTRODUCTION}

The rapid growth of renewable and energy storage resources (ESRs) is transforming real-time power system operations, introducing increased uncertainty, intertemporal constraints, and opportunity costs that existing electricity market pricing mechanisms struggle to capture \cite{EIA2025Battery,CAISO2025BatteryReport}. Grid operators, such as California Independent System Operator (ISO) and New York ISO, increasingly rely on rolling-window economic dispatch to coordinate ESRs operations and shift energy from low-price surplus periods to high-price scarcity periods. However, rolling-window dispatch introduces intertemporal trade-offs: ESR’s current charging and discharging decisions affect its future state-of-charge (SOC) and its ability to respond to future prices. As a result, forward-looking opportunity costs arise and are not fully captured by the existing pricing mechanism, namely locational marginal price (LMP), particularly under forecast uncertainty from renewable variability \cite{Hogan:20, Zhao2019multi}.

This limitation of LMP leads to out-of-merit dispatch, where following the operator’s rolling-window dispatch signals may result in negative profits for ESRs \cite{CAISO2025StorageDesign, CAISO2024StorageBCR}. For example, an ESR may receive a dispatch signal to charge at high prices in anticipation of future scarcity, even when the current price does not compensate for its intertemporal opportunity cost. If the anticipated scarcity together with high prices does not materialize due to the operator's imperfect forecast, ESR faces losses from out-of-merit dispatch and may deviate from operator’s dispatch instructions to avoid deficits. To preserve dispatch-following incentives and maintain grid reliability, operators provide out-of-market uplift payments—bid cost recovery (BCR)—which compensate resources ex post when in-market revenues fall short of ESRs' bid-in costs in the real-time dispatch.

While BCR removes undercompensation and incentivizes resources to follow the realtime dispatch, it introduces discriminative out-of-market payments that reduce price transparency and raise strategic behavior concerns \cite{CAISO2024StorageBCR, Guo&Chen&Tong:21TPS}. ESRs may inflate bids to increase expected BCR uplifts, as their costs are largely opportunity-based—dependent on future prices, SOC, and operational strategies—and thus difficult for operators to verify. In 2023-2024, ESRs  in California ISO (CAISO) accounted for roughly 10–11\% of total BCR uplifts to all generators \cite{CAISO2025BatteryReport}. In extreme cases, individual ESRs have received BCR uplifts ranging from \$100,000 to \$240,000 in a single day.\footnote{This example from March 2022 illustrates the outcomes when ESR strategically bids to exploit BCR. Operators continue to monitor BCR closely and have expressed concern about exploitation, although no extreme events have been reported from 2023 to 2026.} Because unwarranted BCR payments are generally recovered from load, strategic bidding by ESRs shifts costs onto ratepayers and reduces social welfare. 

Motivated by market transparency concerns caused by BCR, this paper analyzes the drivers of BCR, establishes the impossibility result and characterizes conditions for zero-BCR uniform pricing, and studies uniform pricing mechanisms for real-time rolling-window dispatch with ESRs.

\vspace{-0.5em}
\subsection{Related work}

Pricing in multi-interval real-time markets has received considerable attention due to the out-of-market uplift issues introduced by intertemporal operational constraints \cite{Hogan:20, Zhao2019multi}. Temporal locational marginal pricing (TLMP) \cite{Guo&Chen&Tong:21TPS} and its stochastic extensions \cite{ PricingESR, Werner2023} have been shown to eliminate such uplifts even under forecast uncertainty. However, TLMP introduces nonuniform price adders to LMP that depend on whether SOC intertemporal constraints are binding, reducing pricing transparency. As a result, uniform pricing remains prevalent in real-time markets because of its transparency.

To reduce uplift payments while maintaining price transparency, several uniform pricing extensions of LMP have been proposed. Price-preserving multi-interval pricing (PMP) and constraint-preserving multi-interval pricing (CMP) extend LMP to multi-interval settings \cite{Hogan:20, CMP}, and stochastic PMP further incorporates demand uncertainty \cite{ChoPapavasiliou2022}. Other approaches, such as max dispatch cost pricing (MDCP) and max temporal LMP (MTLMP), aim to eliminate BCR uplifts based on marginal production costs \cite{Chen26ramping}. However, these methods primarily focus on conventional generators and do not address the distinct challenges posed by ESRs.

The participation of ESRs alters the structure of the real-time market pricing problem. Since ESRs can both charge and discharge, a single uniform energy price must simultaneously serve as a purchase and sale price, creating inherent tension under intertemporal SOC constraints \cite{PricingESR}. More fundamentally, unlike conventional generators whose costs are often tied to observable marginal production costs like fuel costs from gas or coal, ESR costs are largely intertemporal opportunity costs that depend on future prices, SOC, and operating strategy \cite{PricingESR, CAISO2024StorageBCR}. These costs make it difficult for grid operators to verify the truthful bidding of ESRs and make ESRs susceptible to strategic bidding.

The out-of-market uplift payment issues and the difficulty in identifying ESR strategic bidding behaviors motivate our uniform-pricing research that directly targets BCR reduction for ESRs in real-time power systems and electricity markets. 

%\subsection{Summary of Contributions} 

%The main contributions of this paper are summarized as follows. First, we propose a uniform-pricing framework for real-time rolling-window dispatch with ESRs that selects a single settlement price in each interval by minimizing total bid cost recovery (BCR) subject to revenue adequacy for cleared participants. We also introduce a regularized variant that trades off BCR minimization against proximity to a benchmark marginal price, providing explicit control over the balance between uplift reduction and price consistency. 

%Second, we characterize when a zero-BCR uniform price exists. In particular, a zero-BCR price exists if and only if the maximum marginal cost among active generators and discharging ESRs is no greater than the minimum marginal cost among active charging ESRs. This condition is often violated in practice; for example, when low-cost charging coexists with high-cost generation or discharging, strictly positive BCR becomes unavoidable.

%Third, we identify a window-level indicator $T_{i,t'}$ associated with binding intertemporal SOC constraints and show empirically that positive BCR arises only when this indicator is active under uniform pricing. Monte Carlo simulations under forecast uncertainty further show that the proposed schemes significantly reduce total BCR and demand payments relative to LMP while maintaining zero merchandising surplus.

\subsection{Summary of Contributions} 

We study uniform pricing methods for real-time rolling-window dispatch with ESRs and investigate the origins of out-of-market BCR uplift payments. Renewable uncertainty makes this problem more complex. Our main contributions are two-fold.

First, we characterize conditions under which a uniform price achieves zero BCR for all dispatched participants (Proposition~\ref{prop:zero_bcr}). In particular, a zero-BCR uniform price exists if and only if the maximum marginal cost among dispatched generators and discharging ESRs is no greater than the minimum marginal cost among dispatched charging ESRs. Unfortunately, this condition is often violated in practice; for example, when low-cost charging coexists with high-cost generation or discharging. Therefore, we derive an impossibility result in Corollary~\ref{cor:impossibility} of removing BCR uplifts for ESRs with uniform pricing. Following this, we develop an optimal uniform price minimizing BCR (UP-BCR) in the real-time market.
%subject to revenue adequacy of the grid operator.

Second, we conduct empirical studies under forecast uncertainty and demonstrate that UP-BCR significantly reduces total BCR and demand payments compared to LMP, while maintaining zero merchandising surplus for the grid operator. We further introduce an intertemporal coupling indicator associated with binding SOC constraints and empirically show that positive BCR is observed only when this indicator is active under uniform pricing. This provides insight that BCR is fundamentally driven by intertemporal coupling in rolling-window dispatch.

% ============================================================
%  Section II: Rolling-Window Economic Dispatch
%  CDC Conference Paper  –– v7 (clean constraints)
%  Required packages: amsmath
% ============================================================

\section{Rolling-Window Economic Dispatch}
\label{sec:rolling_window}

We consider a copper-plate model over a dispatch horizon $\mathcal{H} := \{1,\ldots,T\}$. Rolling-window economic dispatch is modeled here with each window start time $t' \in \mathcal{H}$, and the system operator solves a
$W$-period look-ahead economic dispatch problem over the rolling window
$\mathcal{H}_{t'} := \{t',\, t'+1,\,\ldots,\, t'+W-1\}$.
Let $\mathcal{N}:=\{1,\ldots,N\}$ denote the set of ESRs
and $\mathcal{M}:=\{1,\ldots,M\}$ denote the set of conventional generation units.
For each $t \in \mathcal{H}_{t'}$, the operator is given a demand forecast $\hat{d}_t$,
and we assume $\hat{d}_{t'} = d_{t'}$ where $d_{t'}$ is the realized demand.
Let $c_j^\mathrm{G}$ denote the bid-in cost of generator $j \in \mathcal{M}$,
and $c_i^\mathrm{D}$, $c_i^\mathrm{C}$ denote the bid-in costs of ESR $i \in \mathcal{N}$
for discharging and charging, respectively.

For each rolling window $\mathcal{H}_{t'}$, the operator solves a $W$-interval look-ahead economic dispatch (ED) by jointly
optimizing generator outputs $g_{j,t}$, ESR charging/discharging powers
$(g^\mathrm{C}_{i,t}, g^\mathrm{D}_{i,t})$, and SOC $e_{i,t}$ over $t\in\mathcal{H}_{t'}$,
given the initial SOC $s^*_{i,t'-1}$ from the previous interval.
Unless otherwise stated, all constraints apply to every
$j\in\mathcal{M}$, $i\in\mathcal{N}$, and $t\in\mathcal{H}_{t'}$.
The rolling-window ED is denoted by $\mathcal{G}^{\mathrm{RED}}_{t'}$:
\begin{subequations}
\label{eq:RED}
\begin{align}
\min \quad
& \sum_{t\in\mathcal{H}_{t'}}
\!\left(
\sum_{j=1}^{M} c^\mathrm{G}_j\, g_{j,t}
+ \sum_{i=1}^{N}\big(c^\mathrm{D}_i\, g^\mathrm{D}_{i,t} - c^\mathrm{C}_i\, g^\mathrm{C}_{i,t}\big)
\right)
\label{eq:RED_obj}
\\[3pt]
\text{s.t.}\quad
& \sum_{j=1}^{M} g_{j,t} + \sum_{i=1}^{N}\big(g^\mathrm{D}_{i,t}-g^\mathrm{C}_{i,t}\big)
= \hat{d}_t
\;:\; (\lambda_t)
\label{eq:RED_balance}
\\[3pt]
& e_{i,t} = e_{i,t-1} + \eta^\mathrm{C}_i g^\mathrm{C}_{i,t} - \tfrac{1}{\eta^\mathrm{D}_i} g^\mathrm{D}_{i,t},
\quad t > t'
\;:\; (\phi_{i,t})
\label{eq:RED_soc_dyn}
\\[3pt]
& e_{i,t'} = s^*_{i,t'-1} + \eta^\mathrm{C}_i g^\mathrm{C}_{i,t'} - \tfrac{1}{\eta^\mathrm{D}_i} g^\mathrm{D}_{i,t'}
\;:\; (\phi_{i,t'})
\label{eq:RED_soc_init}
\\[3pt]
& \underline{e}_i \le e_{i,t} \le \bar{e}_i
\;:\; (\underline\rho_{i,t},\, \bar\rho_{i,t})
\label{eq:RED_soc_bounds}
\\[3pt]
& 0 \le g_{j,t} \le \bar{g}_j
\;:\; (\underline\mu^\mathrm{G}_{j,t},\, \bar\mu^\mathrm{G}_{j,t})
\label{eq:RED_gen_bounds}
\\[3pt]
& 0 \le g^\mathrm{C}_{i,t} \le \bar{g}^\mathrm{C}_i
\;:\; (\underline\mu^\mathrm{C}_{i,t},\, \bar\mu^\mathrm{C}_{i,t})
\label{eq:RED_charge_bounds}
\\[3pt]
& 0 \le g^\mathrm{D}_{i,t} \le \bar{g}^\mathrm{D}_i
\;:\; (\underline\mu^\mathrm{D}_{i,t},\, \bar\mu^\mathrm{D}_{i,t})
\label{eq:RED_discharge_bounds}
\end{align}
\end{subequations}

\noindent
where $\lambda_t$ is the dual variable for the power-balance constraint
\eqref{eq:RED_balance}.
For simplicity, we adopt a copper-plate model and ignore the ramping constraints.
For each ESR $i$, $\phi_{i,t}$ is the dual variable for the SOC dynamics
\eqref{eq:RED_soc_dyn}--\eqref{eq:RED_soc_init}, capturing the shadow price of
stored energy and coupling decisions across intervals. Storage charging and discharging efficiencies are denoted by $\eta_i^{\mathrm D},\eta_i^{\mathrm C}\in(0,1]$.

\begin{remark}[No simultaneous charging and discharging]
\label{rem:no_simultaneous}
We assume that for each ESR $i$, the bid parameters satisfy
$c_i^\mathrm D > c_i^\mathrm C/(\eta_i^\mathrm D \eta_i^\mathrm C)$, which ensures that the complementarity condition $g_{i,t}^\mathrm C \, g_{i,t}^\mathrm D = 0$ holds at the optimality of~\eqref{eq:RED} \cite{PricingESR}.
\end{remark}

Although \eqref{eq:RED} optimizes over the entire window $\mathcal{H}_{t'}$,
only the dispatch at the binding interval $t'$ is implemented.
Here, the binding interval $t'$ refers to the first interval of the rolling-window $\mathcal{H}_{t'}$.
Let superscript $*$ denote an optimal solution of \eqref{eq:RED}.
The realized rolling-window dispatch signals are
\begin{equation}\label{eq:dispatch_signals}
g^{\mathrm{RED}}_{j,t'} := g^*_{j,t'}, \quad
g^{\mathrm{RED \text{-} C}}_{i,t'} := g_{i,t'}^{\mathrm{C}*}, \quad
g^{\mathrm{RED \text{-} D}}_{i,t'} := g_{i,t'}^{\mathrm{D}*}.
\end{equation}
and the SOC is updated as $s^*_{i,t'} := e^*_{i,t'}$ for all $i\in\mathcal{N}$.
At the next window start $t'+1$, the operator receives updated forecasts
$\{\hat{d}_t\}_{t\in\mathcal{H}_{t'+1}}$ and resolves \eqref{eq:RED} over $\mathcal{H}_{t'+1}$ using $s^*_{i,t'}$ as the initial SOC.

Repeating this for $t'=1,\ldots,T$ yields the rolling-window dispatch policy
$\mathcal{G}^{\mathrm{RED}} := \{\mathcal{G}^{\mathrm{RED}}_{t'}\}_{t'=1}^{T}$.

% ============================================================
%  Section III: Pricing and Settlement
%  Section IV: Performance Metrics
% ============================================================

% ------------------------------------------------------------
\section{Pricing and Settlement}
\label{sec:pricing}

Given the rolling-window dispatch $\mathcal{G}^{\mathrm{RED}}$ from Section~\ref{sec:rolling_window}, 
we first introduce benchmark marginal pricing rules (LMP, TLMP, and MTLMP \cite{Guo&Chen&Tong:21TPS, PricingESR}). Then, we define BCR and finally present a uniform-pricing design minimizing BCR uplift, inspired by \cite{Chen26ramping}.
% --- III-A -------------------------------------------------
\vspace{-0.9em}
\subsection{Benchmark Pricing Rules}

\paragraph{LMP}
The locational marginal price (LMP) at interval \( t \) is given by the optimal dual variable 
associated with the power-balance constraint \eqref{eq:RED_balance} in the binding interval:
\begin{equation}\label{eq:lmp}
  \pi^{\mathrm{LMP}}_t := \lambda^*_t .
\end{equation}

\paragraph{TLMP}
The temporal locational marginal price (TLMP) extends LMP by incorporating intertemporal coupling through the SOC dynamics and defines ESR-specific marginal prices for discharge and charge \cite{PricingESR}. 
For ESR \( i \) at the binding interval \( t \), TLMP for discharge and charge are  
\begin{equation}\label{eq:tlmp}
\tilde \pi^{\mathrm{D}}_{i,t}
:= \lambda^*_t - \frac{\phi^*_{i,t}}{\eta^{\mathrm{D}}_i},
\qquad
\tilde \pi^{\mathrm{C}}_{i,t}
:= \lambda^*_t - \eta^{\mathrm{C}}_i\,\phi^*_{i,t},
\end{equation}
where \( \phi^*_{i,t} \) is the optimal dual variable associated with the SOC-dynamics constraints \eqref{eq:RED_soc_dyn}--\eqref{eq:RED_soc_init}, capturing the intertemporal opportunity cost of storage.

\paragraph{MTLMP}
We generalize the max temporal locational marginal price (MTLMP) for generators in \cite{Chen26ramping} to model \eqref{eq:RED} with ESRs and generators. 
MTLMP is a uniform price in \cite{Chen26ramping} for 
conventional generators.\footnote{A uniform price means that all units---generators, loads, and storage---pay and are paid the same price.} 
With ESRs, the extension of MTLMP here uses separate charging and discharging prices denoted by $\bar\pi^{\mathrm{D}}_t$ and
$\bar\pi^{\mathrm{C}}_t$,
and is therefore nonuniform.
MTLMP defines clipped discharge and charge prices by taking the worst-case TLMP across ESRs. We use the following clipped variant:
\begin{equation}\label{eq:mtlmp}
\begin{aligned}
  \bar\pi^{\mathrm{D}}_t
  &:= \max_{i\in \mathcal{N}}
      \Bigl(\lambda^*_t - \phi^*_{i,t}/\eta^{\mathrm{D}}_i,\ \lambda^*_t\Bigr), \\
  \bar\pi^{\mathrm{C}}_t
  &:= \min_{i \in \mathcal{N}}
      \Bigl(\lambda^*_t - \eta^{\mathrm{C}}_i\phi^*_{i,t},\ \lambda^*_t\Bigr).
\end{aligned}
\end{equation}

% --- III-B -------------------------------------------------
\vspace{-0.9em}
\subsection{Dispatch-Following Profit and BCR}\label{sec:BCR}

For a given settlement rule, let \( \pi_{i,t}^\mathrm{D} \) and \( \pi_{i,t}^\mathrm{C} \) denote the discharge and charge prices for ESR \( i \) in interval \( t \), and let \( \pi_{j,t} \) denote the settlement price for generator \( j \).
Given the rolling-window dispatch defined in~\eqref{eq:dispatch_signals}, the dispatch-following profit of ESR \( i \) in interval \( t \) is

\begin{equation}\label{eq:profit_esr}
\Pi^{\mathrm{ESR}}_{i,t}(\pi_{i,t}^\mathrm{D},\pi_{i,t}^\mathrm{C})
:= (\pi_{i,t}^\mathrm{D} - c_i^\mathrm{D}) g^{\mathrm{RED \text{-} D}}_{i,t}
- (\pi_{i,t}^\mathrm{C}-c_i^\mathrm{C}) g^{\mathrm{RED \text{-} C}}_{i,t}.
\end{equation}
The dispatch-following profit of generator \( j \) in interval \( t \) is  
\begin{equation}\label{eq:profit_gen}
\Pi^{\mathrm{G}}_{j,t}(\pi_{j,t})
:= (\pi_{j,t} - c_j^\mathrm{G})g^{\mathrm{RED}}_{j,t}.
\end{equation}

The per-interval bid cost recovery (BCR) is the make-whole payment required to ensure non-negative profit:
\begin{equation}\label{eq:bcr_def}
\begin{aligned}
  \mathcal{B}_{i,t}^\mathrm{ESR}
  &:= \max\{0,-\Pi^{\mathrm{ESR}}_{i,t}(\pi_{i,t}^\mathrm{D},\pi_{i,t}^\mathrm{C})\}, \\
  \mathcal{B}_{j,t}^\mathrm{G}
  &:= \max\{0,-\Pi^{\mathrm{G}}_{j,t}(\pi_{j,t})\}.
\end{aligned}
\end{equation}

% --- III-C -------------------------------------------------
\subsection{Uniform Pricing via BCR Minimization (UP-BCR)}
\label{sec:up_bcr}

We design a uniform pricing rule that determines a single price $\pi_t\in\mathbb{R}$ per interval to minimize total BCR uplift. 
%\footnote{\textcolor{orange}{UP-BCR can be viewed as a single-bus, ESR-oriented
%specialization of the uniform price with minimum uplift (UPMU)
%framework in~\cite{Chen26ramping}: here the merchandising-surplus
%constraint is redundant since merchandising surplus is identically zero
%under single-bus uniform pricing (Section~\ref{subsec:MS}), and the
%LMP-deviation weight is set to zero to focus solely on BCR
%minimization.}}
\footnote{UP-BCR is a copper-plate, ESR-oriented specialization of
UPMU~\cite{Chen26ramping}, with the LMP-deviation weight set to zero.
The merchandising-surplus constraint is redundant under copper-plate
uniform pricing (Section~\ref{subsec:MS}).}
Under uniform pricing, a uniform price $\pi_t$ is used for both charging and discharging, i.e., \(\pi^\mathrm{D}_{i,t} = \pi^\mathrm{C}_{i,t} = \pi_t, \forall i,t.\)
Substituting into \eqref{eq:profit_esr}, the ESR dispatch-following profit becomes a function of $\pi_t$ only. We introduce scalar variables $\mathcal{B}^{\mathrm{ESR}}_{i,t}, \mathcal{B}^{\mathrm{G}}_{j,t} \in \mathbb{R}_{+}$ to represent the make-whole payments
\footnote{Variables $\mathcal{B}^{\mathrm{ESR}}_{i,t}$ and
$\mathcal{B}^{\mathrm{G}}_{j,t}$ are introduced in~\eqref{eq:UP1}.
For simplicity, we use the same notation in~\eqref{eq:bcr_def} to define BCR.}.

%\begin{equation}\label{eq:bcr_vars}
%\begin{aligned}
%\mathcal{B}^{\mathrm{ESR}}_{i,t} \ge -\Pi^{\mathrm{ESR}}_{i,t}(\pi_t), \quad \mathcal{B}^{\mathrm{ESR}}_{i,t}\ge 0,\quad \forall i,t, \\
%\mathcal{B}^{\mathrm{G}}_{j,t} \ge -\Pi^{\mathrm{G}}_{j,t}(\pi_t), \quad \mathcal{B}^{\mathrm{G}}_{j,t}\ge 0,\quad \forall j,t.
%\end{aligned}
%\end{equation}

The UP-BCR price is defined by the optimal price $\pi_t^*$ of BCR minimization:
% \begin{subequations}\label{eq:UP1}
% \begin{align}
% \min_{\{\pi_t\},\,\{\mathcal{B}^{\mathrm{ESR}}_{i,t}\},\,\{\mathcal{B}^{\mathrm{G}}_{j,t}\}}
% \quad &
% \sum_{t=1}^{T}\left(
% \sum_{i=1}^{N} \mathcal{B}^{\mathrm{ESR}}_{i,t}
% + \sum_{j=1}^{M} \mathcal{B}^{\mathrm{G}}_{j,t}
% \right) \label{eq:UP1_obj} \\
% \text{s.t.}\quad
% & \textcolor{orange}{\eqref{eq:bcr_vars}.} \label{eq:UP1_nn}
% \end{align}
% \end{subequations}
\begin{subequations}\label{eq:UP1}
\begin{align}
\min_{\substack{\{\pi_t\},\,\{\mathcal{B}^{\mathrm{G}}_{j,t}\},\\ \{\mathcal{B}^{\mathrm{ESR}}_{i,t}\}}}   
% \{\pi_t\},\,\{\mathcal{B}^{\mathrm{ESR}}_{i,t}\},\,\{\mathcal{B}^{\mathrm{G}}_{j,t}\}
\quad &
\sum_{t=1}^{T}\left(
\sum_{i=1}^{N} \mathcal{B}^{\mathrm{ESR}}_{i,t}
+ \sum_{j=1}^{M} \mathcal{B}^{\mathrm{G}}_{j,t}
\right) \label{eq:UP1_obj} \\
\text{s.t.}\quad
& %\eqref{eq:bcr_vars}.
\mathcal{B}^{\mathrm{ESR}}_{i,t} \ge -\Pi^{\mathrm{ESR}}_{i,t}(\pi_t),\quad
\mathcal{B}^{\mathrm{ESR}}_{i,t}\ge 0,\quad \forall i,t,
\label{eq:UP1_esr}\\
&\mathcal{B}^{\mathrm{G}}_{j,t} \ge -\Pi^{\mathrm{G}}_{j,t}(\pi_t),\quad
\mathcal{B}^{\mathrm{G}}_{j,t}\ge 0,\quad \forall j,t.
\label{eq:UP1_nn}
\end{align}
\end{subequations}

\begin{remark}
Although \eqref{eq:UP1} is written in a horizon-level form for compactness, the objective and constraints are separable across intervals. In particular, for each interval $t$, the UP-BCR price $\pi^*_t$ depends only on the realized dispatch quantities $g^{\mathrm{RED}}_{j,t}$, $g^{\mathrm{RED\text{-}D}}_{i,t}$, and $g^{\mathrm{RED\text{-}C}}_{i,t}$ defined by~\eqref{eq:dispatch_signals} at that interval. Therefore, \eqref{eq:UP1} is equivalent to solving $T$ independent single-interval problems, and can be implemented online at each interval.
\end{remark}

Since the dispatch-following profit functions in~\eqref{eq:profit_esr} and~\eqref{eq:profit_gen}
are affine in $\pi_t$, the UP-BCR problem~\eqref{eq:UP1} is a linear program. This UP-BCR problem is equivalent to a convex optimization minimizing BCR in the objective directly. See the appendix for the reformulation and a proof of equivalence.

% ============================================================
% ============================================================
%  Section V: Impossibility of Zero-BCR Uniform Pricing
% ============================================================
 
\section{Impossibility of Zero-BCR Uniform Pricing}

Out-of-market uplift payments such as BCR reduce price transparency and can incentivize strategic bidding behavior. A natural question is therefore whether a transparent uniform price can eliminate BCR for all resources. For generators and discharging ESRs, zero BCR requires the uniform price to be above their bid costs; while for charging ESRs, zero BCR requires the price to be below their charging bid costs. Therefore, once both charging and discharging actions are present, these requirements may conflict. In this section, we characterize the feasibility of zero-BCR uniform pricing and show that, when the resulting price bounds are incompatible, it is impossible to reduce BCR to zero for all resources.
 
To analyze the BCR uplifts for all dispatched market participants, we define the set of dispatched units for a fixed interval~$t$ as
\begin{equation} \label{eq:active_sets}
\begin{aligned}
    \mathcal{G}_t^\mathrm{+} := \{j : g_{j,t}^{\mathrm{RED}} > 0\},       \\
    \mathcal{S}_t^\mathrm{D,+} := \{i : g_{i,t}^{\mathrm{RED\text{-}D}} > 0\},     \\
    \mathcal{S}_t^\mathrm{C,+} := \{k : g_{k,t}^{\mathrm{RED\text{-}C}} > 0\}.
\end{aligned}
\end{equation}

\begin{remark}
\label{rem:nontrivial_case}
If $\mathcal{S}_t^\mathrm{C,+} = \emptyset$, there is no price cap for the uniform market price $\pi_t$, and a zero-BCR price trivially exists by setting
$\pi_t \geq \max\!\left\{\underset{j\in\mathcal{G}_t^\mathrm{+}}{\max} c_j^\mathrm{G},\; \underset{i\in\mathcal{S}_t^\mathrm{D,+}} {\max}c_i^\mathrm{D}\right\}$.
Symmetrically, if $\mathcal{G}_t^\mathrm{+} = \mathcal{S}_t^\mathrm{D,+} = \emptyset$, any
$\pi_t \leq \underset{k\in\mathcal{S}_t^\mathrm{C,+}}{\min} c_k^\mathrm{C}$
achieves zero BCR.
Therefore, for a fixed interval $t$, we restrict attention to the nontrivial case in which
$(\mathcal{G}_t^\mathrm{+}\cup\mathcal{S}_t^\mathrm{D,+})\neq\emptyset$
and $\mathcal{S}_t^\mathrm{C,+}\neq\emptyset$.
\end{remark}

For a given realized dispatch at interval $t$, zero-BCR feasibility asks
whether there exists a uniform price $\pi_t$ such that all dispatched participants
earn nonnegative dispatch-following profit. Because the same price applies to
both charging and discharging, such a price exists only when the corresponding
lower and upper bounds on $\pi_t$ are compatible.

\begin{proposition}[Zero-BCR Feasibility Condition]
\label{prop:zero_bcr}
For a fixed interval $t$ with nonempty $\mathcal{G}_t^\mathrm{+}, \mathcal{S}_t^\mathrm{D,+}$,
and $\mathcal{S}_t^\mathrm{C,+}$, a uniform price $\pi_t$ achieves zero BCR for all dispatched
participants if and only if
\begin{equation}
  \label{eq:feasibility}
  \max\!\left\{
    \max_{j \in \mathcal{G}_t^\mathrm{+}} c_j^\mathrm{G},\;
    \max_{i \in \mathcal{S}_t^\mathrm{D,+}} c_i^\mathrm{D}
  \right\}
  \leq
  \min_{k \in \mathcal{S}_t^\mathrm{C,+}} c_k^\mathrm{C}.
\end{equation}
\end{proposition}

\begin{corollary}[Impossibility of Zero Total BCR]
\label{cor:impossibility}
For a fixed interval~$t$ with nonempty $\mathcal{G}_t^\mathrm{+}, \mathcal{S}_t^\mathrm{D,+}$,
and $\mathcal{S}_t^\mathrm{C,+}$, when condition~\eqref{eq:feasibility} fails, zero BCR is infeasible for the UP-BCR problem in~\eqref{eq:UP1}.
%
%\begin{equation}
  %\min_{\pi_t}
 % \left(
    %\sum_{j\in\mathcal{G}_t^\mathrm{+}} \mathcal B_{j,t}^\mathrm{G}
  %  + %\sum_{i\in\mathcal{S}_t^\mathrm{D,+}} \mathcal B_{i,t}^\mathrm{D}
  %  + \sum_{k\in\mathcal{S}_t^\mathrm{C,+}} \mathcal B_{k,t}^\mathrm{C}
 % \right)
 % > 0.
%\end{equation}
\end{corollary}

Proofs of Proposition~\ref{prop:zero_bcr} and Corollary~\ref{cor:impossibility} are in the appendix. Condition~\eqref{eq:feasibility} fails when
low-cost charging coexists with high-cost generation
or discharging, creating conflicting price requirements.
Such a conflict does not arise under single-interval real-time dispatch.
It emerges under multi-interval rolling-window dispatch with binding
state-of-charge constraints. Section~\ref{sec:toy} gives a concrete example
illustrating Corollary~\ref{cor:impossibility}.
% ============================================================
%  Section VI: Toy Example
% ============================================================
 
\section{Toy Example for Impossibility Result}
\label{sec:toy}
 
We illustrate Corollary~\ref{cor:impossibility} with a concrete configuration of dispatched units in which one ESR
discharges while another ESR charges, demonstrating that
condition~\eqref{eq:feasibility} can fail.

\paragraph{Setup}
Consider a copper-plate system with a rolling window of length $W=2$.
The realized demand at the binding interval is $d_t=22~\mathrm{MW}$,
and the forecast demand for the second interval is
$\hat d_{t+1}=27~\mathrm{MW}$.
The ESR bids are
$(c_1^{\mathrm D},c_1^{\mathrm C})=(9,5)~\$/\mathrm{MWh}$ and
$(c_2^{\mathrm D},c_2^{\mathrm C})=(20,1)~\$/\mathrm{MWh}$,
and the generator bid is $c_1^{\mathrm G}=10~\$/\mathrm{MWh}$.
The generator capacity is $\bar g_1=20~\mathrm{MW}$.
For both ESRs,
$\bar g_i^{\mathrm C}=\bar g_i^{\mathrm D}=4~\mathrm{MW}$,
$\underline e_i=2~\mathrm{MWh}$,
$\bar e_i=12~\mathrm{MWh}$, and
$\eta_i^{\mathrm C}=\eta_i^{\mathrm D}=1$.
The initial SOCs are
$s^*_{1,t-1}=10~\mathrm{MWh}$ and
$s^*_{2,t-1}=4~\mathrm{MWh}$.

Solving \eqref{eq:RED} over the two-interval rolling window gives
the optimal dispatch for $t$ and $t+1$, as summarized in
Table~\ref{tab:toy_dispatch}.

{\setlength{\textfloatsep}{6pt}
\begin{table}[t]
\centering
\caption{Dispatch results for the toy example.}
\label{tab:toy_dispatch}
\setlength{\tabcolsep}{3pt}
\begin{tabular}{c|cc|ccccc}
\hline
 & \multicolumn{2}{c|}{SOC ($\mathrm{MWh}$)}
 & \multicolumn{5}{c}{Dispatch ($\mathrm{MW}$)} \\
 & ESR 1 & ESR 2
 & $g_1^*$ & $g_1^{\mathrm D*}$ & $g_1^{\mathrm C*}$
 & $g_2^{\mathrm D*}$ & $g_2^{\mathrm C*}$ \\
\hline
$t-1$ & $10$ & $4$ & -- & -- & -- & -- & -- \\
$t$   & $6$  & $5$ & $19$ & $4$ & $0$ & $0$ & $1$ \\
$t+1$ & $2$  & $2$ & $20$ & $4$ & $0$ & $3$ & $0$ \\
\hline
\end{tabular}
\end{table}
}

By the definition of the set of dispatched units in \eqref{eq:active_sets}, we have
$\mathcal{G}_t^\mathrm{+} = \{1\}$,
$\mathcal{S}_t^\mathrm{D,+} = \{1\}$, and $\mathcal{S}_t^\mathrm{C,+} = \{2\}$.
Checking condition~\eqref{eq:feasibility}, where all bid costs are in~$\$/\mathrm{MWh}$, we obtain
\begin{equation*}
  \max\!\left\{
    \underbrace{c_1^\mathrm{G}}_{10},\,
    \underbrace{c_1^\mathrm{D}}_{9}
  \right\} = 10 > 1 = 
  \underbrace{c_2^{\mathrm C}}_{
    \substack{
      \min\; c_k^{\mathrm C}\\
      \scriptscriptstyle k\in\mathcal S_t^{\mathrm C,+}
    }}.
\end{equation*}
The condition fails; by Corollary~\ref{cor:impossibility},
no uniform price can achieve zero BCR.

\paragraph{Impossibility of nonnegative profit and zero BCR}
Let $\pi$ denote the uniform price, i.e.,
$\pi_{i,t}^{\mathrm D}=\pi_{i,t}^{\mathrm C}=\pi$.
From \eqref{eq:profit_esr}--\eqref{eq:profit_gen}, zero BCR for all
dispatched participants requires nonnegative profits. Since
$g^{\mathrm{RED\text{-}D}}_{1,t}$,
$g^{\mathrm{RED\text{-}C}}_{2,t}$, and
$g^{\mathrm{RED}}_{1,t}$ are positive, we have
\begin{align*}
\Pi_{1,t}^{\mathrm{ESR}}(\pi,\pi)
&= (\pi-9)g^{\mathrm{RED\text{-}D}}_{1,t}\geq0
\;\Rightarrow\; \pi\geq9,\\
\Pi_{2,t}^{\mathrm{ESR}}(\pi,\pi)
&= (1-\pi)g^{\mathrm{RED\text{-}C}}_{2,t}\geq0
\;\Rightarrow\; \pi\leq1,\\
\Pi_{1,t}^{\mathrm{G}}(\pi)
&= (\pi-10)g^{\mathrm{RED}}_{1,t}\geq0
\;\Rightarrow\; \pi\geq10.
\end{align*}
These conditions cannot hold simultaneously. Therefore, no uniform
price can achieve zero BCR while supporting nonnegative profits for
all dispatched participants; hence the total BCR is strictly positive,
illustrating Corollary~\ref{cor:impossibility}.

\section{Performance Metrics}
\label{sec:metrics}

The following empirical study evaluates pricing rules using metrics including demand payments, generator revenues, and generator profits. Most importantly, we introduce an intertemporal coupling indicator that reveals the underlying driver of BCR under uniform pricing.

% --- IV-D -------------------------------------------------
\vspace{-0.5em}
\subsection{Intertemporal Coupling Indicator}
\label{subsec:binding_indicator}

To understand when and why BCR arises under uniform pricing, we establish the intertemporal coupling indicator $T_{i,t'}$ in \eqref{eq:window_ind} that captures whether ESR $i$ takes charging/discharging actions and hits binding intertemporal SOC constraints within a rolling window $\mathcal{H}_{t'}$.
\begin{subequations}\label{eq:window_ind}
\begin{align}
T_{i,t'}&:=\mathbf{1}\{\exists\,t\in\mathcal{H}_{t'}:\tilde\delta_{i,t}=1\},\\
\tilde\delta_{i,t}&:=\mathbf{1}\{(\delta^\mathrm{AT}_{i,t}=1)\ \wedge\ (\delta_{i,t}=1)\}.
\end{align}    
\end{subequations}
Here, $\delta^\mathrm{AT}_{i,t}:=\mathbf{1}\{\delta_{g^\mathrm{D}_{i,t}}+\delta_{g^\mathrm{C}_{i,t}}\neq 2\}$ represents the activity indicator, where $\delta_{g^\mathrm{C}_{i,t}}:=\mathbf{1}\{g^\mathrm{C}_{i,t}=0\}$ and $\delta_{g^\mathrm{D}_{i,t}}:=\mathbf{1}\{g^\mathrm{D}_{i,t}=0\}$. It equals zero when ESR $i$ is idle, i.e., when both charging and discharging are zero.

Denote $\delta_{i,t} := \mathbf{1}\{e_{i,t}=\bar e_i\} + \mathbf{1}\{e_{i,t}=\underline e_i\}$ as  the SOC-boundary indicator. Since the SOC cannot simultaneously reach both bounds, we have $\delta_{i,t} \in \{0,1\}$. Specifically, we have $\delta_{i,t} = 1$ when ESR $i$ hits either the 
upper or lower SOC bound at interval $t$, and $\delta_{i,t} = 0$ otherwise.

% --- IV-A -------------------------------------------------
\vspace{-0.5em}
\subsection{Demand Payment}

The total demand payment under a given pricing rule is % defined as
\begin{equation}
\Pi^{\mathrm{DP}} := \sum_{t=1}^{T} \pi_t^{\mathrm{load}} \, d_t,
\end{equation}
where $d_t$ is the realized demand at interval $t$. The load price $\pi_t^{\mathrm{load}}$ is given by
$\pi_t^{\mathrm{load}} = \lambda_t^*$ under LMP and TLMP,
$\pi_t^{\mathrm{load}} = \bar{\pi}_t^\mathrm C$ under MTLMP,
and $\pi_t^{\mathrm{load}} = \pi_t$ under uniform pricing.

% --- IV-B -------------------------------------------------
\subsection{Generator Revenue and Profit}
The total generator revenue under a given pricing rule is %defined as
\begin{equation}
\Pi^{\mathrm{GenPay}} := \sum_{t=1}^{T} \sum_{j=1}^{M} \pi_t^{\mathrm{G}} \, g^{\mathrm{RED}}_{j,t},
\end{equation}
where $g^{\mathrm{RED}}_{j,t}$ is the realized dispatch of generator $j$ at interval $t$, and $\pi_t^{\mathrm{G}}$ denotes the settlement price applied to generators.

The generator price $\pi_t^{\mathrm{G}}$ is given by
$\pi_t^{\mathrm{G}} = \lambda_t^*$ under LMP and TLMP,
$\pi_t^{\mathrm{G}} = \bar{\pi}_t^\mathrm D$ under MTLMP,
and $\pi_t^{\mathrm{G}} = \pi_t$ under uniform pricing.

The total generator profit is defined as
\begin{equation}
\Pi^{\mathrm{GenProfit}} := \sum_{t=1}^{T} \sum_{j=1}^{M} 
\left( \pi_t^{\mathrm{G}} - c_j^{\mathrm{G}} \right) g^{\mathrm{RED}}_{j,t}.
\end{equation}

% --- IV-C -------------------------------------------------
\subsection{Merchandising Surplus}
\label{subsec:MS}

The merchandising surplus (MS) under a given in-market pricing rule is defined as the net revenue of the power system operator:
\begin{equation}
\mathrm{MS} := \Pi^{\mathrm{DP}} - \Pi^{\mathrm{GenPay}} 
+ \sum_{t=1}^{T} \sum_{i=1}^{N} 
\left( \pi_{i,t}^{\mathrm{C}} g_{i,t}^{\mathrm{\mathrm{RED \text{-} C}}} 
- \pi_{i,t}^{\mathrm{D}} g_{i,t}^{\mathrm{\mathrm{RED \text{-} D}}} \right),
\end{equation}
where $(\pi_{i,t}^{\mathrm{C}}, \pi_{i,t}^{\mathrm{D}})$ denote the settlement prices for ESR charging and discharging, respectively, under the corresponding pricing rule. For simplicity, we assume that BCR payments to generators and ESRs are recovered from demand, so the system operator's net position from out-of-market BCR settlements is zero.

Under LMP and uniform pricing, a single price is applied so that $\pi_{i,t}^{\mathrm{D}} = \pi_{i,t}^{\mathrm{C}} = \pi_t$, while under TLMP and MTLMP, the charge and discharge prices may differ as defined in \eqref{eq:tlmp}--\eqref{eq:mtlmp}.

This metric reflects the revenue adequacy of the market
operator. Under uniform pricing with a single price applied
to both load and all participants, the merchandising surplus
is identically zero in a copper-plate model, following power balance in~\eqref{eq:RED_balance} directly.

\section{Simulation Results}

% ------------------------------------------------------------
\subsection{Parameter Settings}
% ------------------------------------------------------------

We consider a copper-plate model with three conventional generators
and two ESRs over a 24-hour horizon ($T=24$), with a four-hour
look-ahead window ($W=4$). We adopt the same ESR power and SOC limits,
efficiencies, and ESR~1 bid as in Section~\ref{sec:toy}.
The generator capacities are $(20,20,100)~\mathrm{MW}$ with marginal
costs $(c_1^{\mathrm G},c_2^{\mathrm G},c_3^{\mathrm G})
=(10,63,100)~\$/\mathrm{MWh}$.
Unlike the toy example, the ESR~2 bid is
$(c_2^{\mathrm D},c_2^{\mathrm C})=(12,8)~\$/\mathrm{MWh}$.

Demand at each interval $t$ is independently drawn from a Gaussian distribution
$\mathcal{N}(\mu_t, 4)$, where $\{\mu_t\}$ follows a typical
daily load profile as illustrated in Fig.~\ref{fig:demand} (top).
Rolling forecasts are subject to lead-time-dependent Gaussian
errors ($\sigma_1 = 1$, $\sigma_2 = 2$, $\sigma_3 = 3$).
We conduct 200 independent Monte Carlo scenarios; all pricing
rules are evaluated under the same realized dispatch to ensure
a fair comparison, and metrics are averaged over all scenarios.

\begin{figure}[t]
  \centering
  \includegraphics[width=0.7\columnwidth]{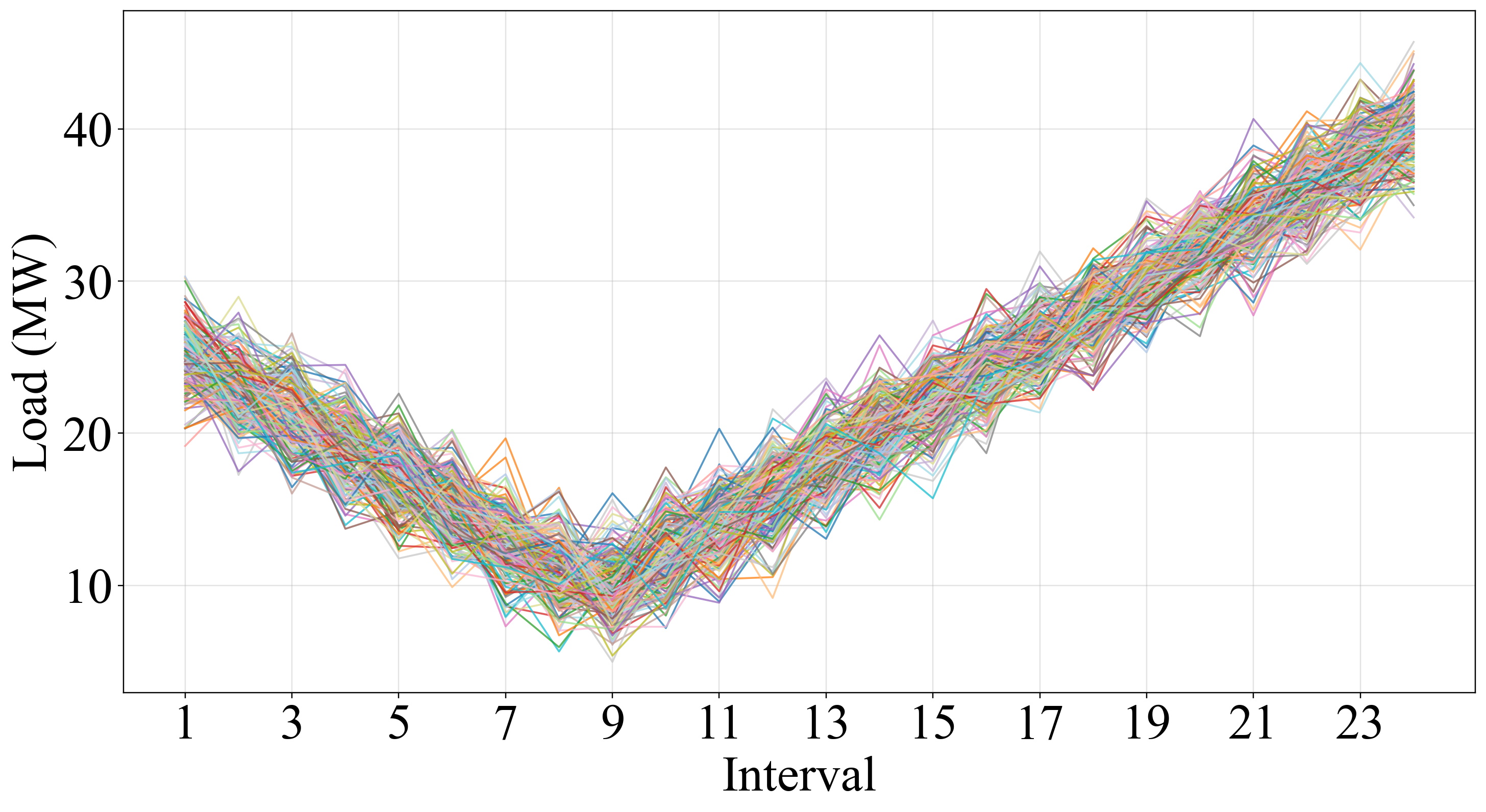}
  
   \includegraphics[width=0.75\columnwidth]{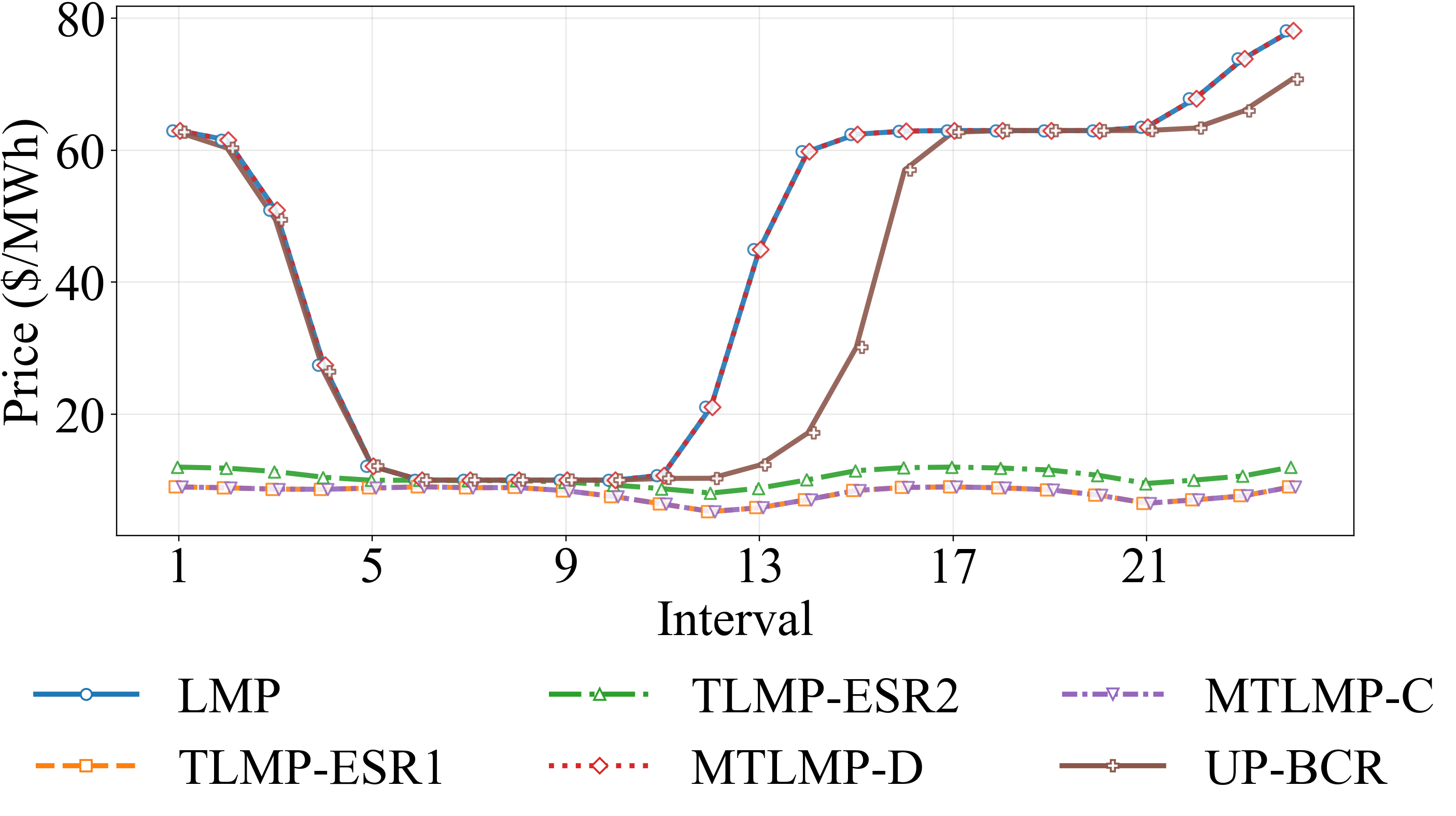}
  \caption{24-hour profiles for load and prices. Top: Daily load profile (200 Monte Carlo scenarios). Bottom: Mean price trajectories per interval.}
  \label{fig:demand}
\end{figure}

% ------------------------------------------------------------
\subsection{BCR Performance and SOC Constraints}
% ------------------------------------------------------------

%\begin{figure}[htbp]
  %\centering
 % \includegraphics[width=0.45\linewidth]{figure/TotalBCR_mean.png} 
%\includegraphics[width=0.45\columnwidth]{figure/Total_BCR_ESR_plus_Gen_total_boxplot.png}
 % \includegraphics[width=0.45\columnwidth]{figure/GenProfit_mean.png}
 % \includegraphics[width=0.45\linewidth]{figure/MS_mean.png}
  % \vspace{-2em}
 % \caption{\small \textcolor{red}{Top left: Mean total BCR per interval. Top right: Distribution of total BCR across 200 scenarios. Bottom left: Mean generator profit per interval. Bottom right: Mean merchandising surplus per interval.}}
  %\label{fig:sim}
%\end{figure}

\begin{figure}[t]
  \centering
  \includegraphics[width=0.48\columnwidth]{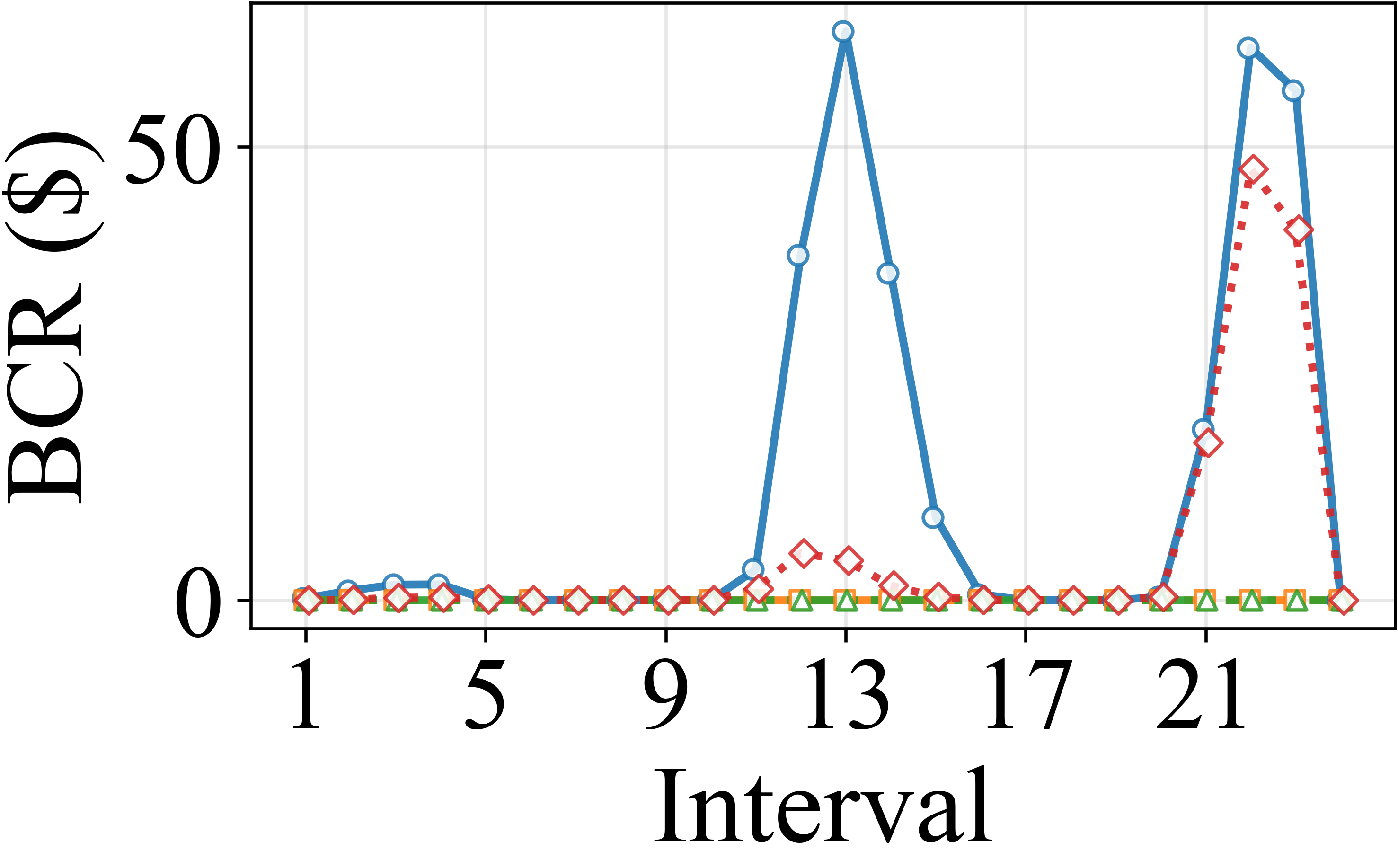}\hfill
  \includegraphics[width=0.48\columnwidth]{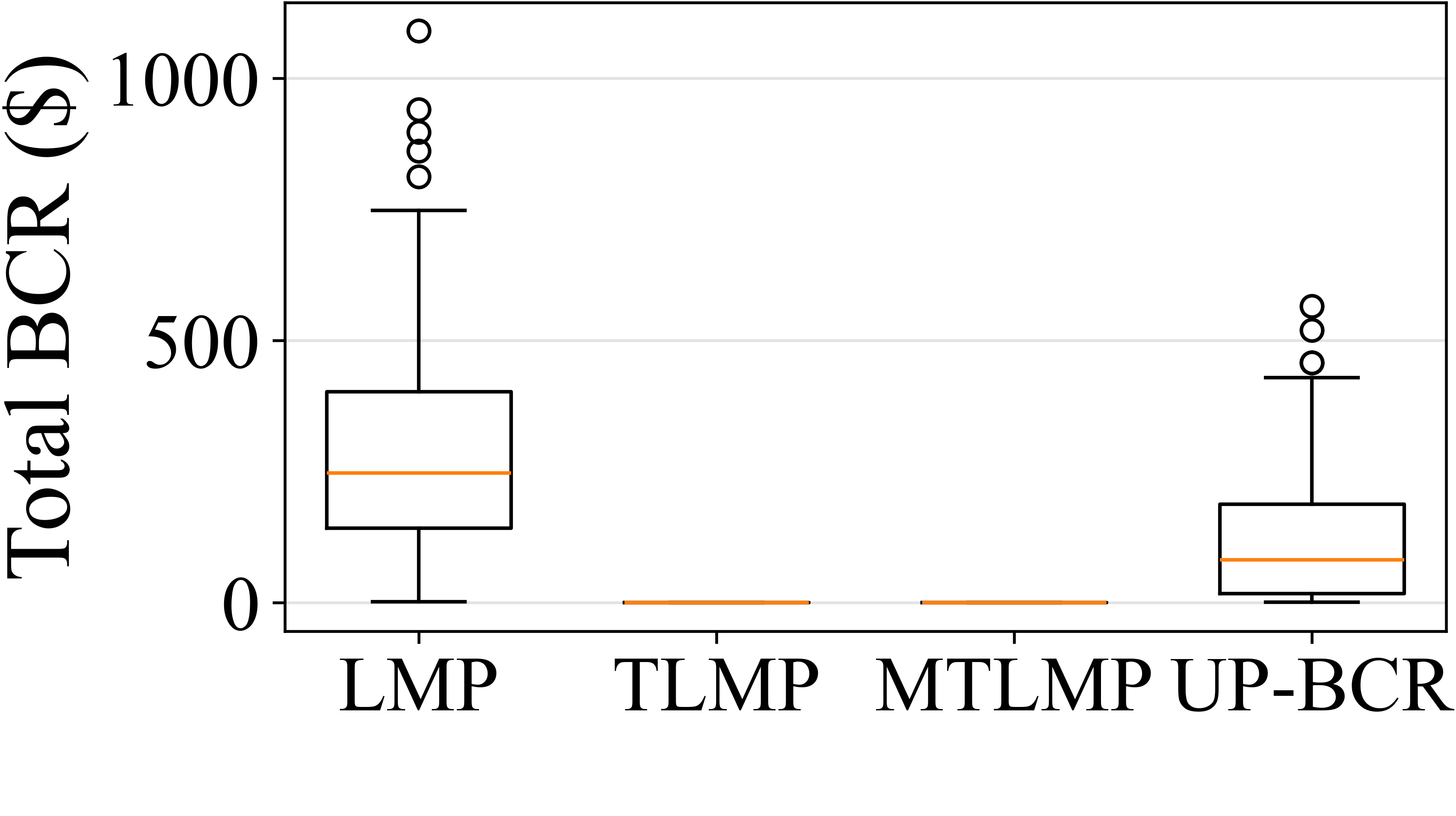}

  \vspace{0.2em}

  \includegraphics[width=0.48\columnwidth]{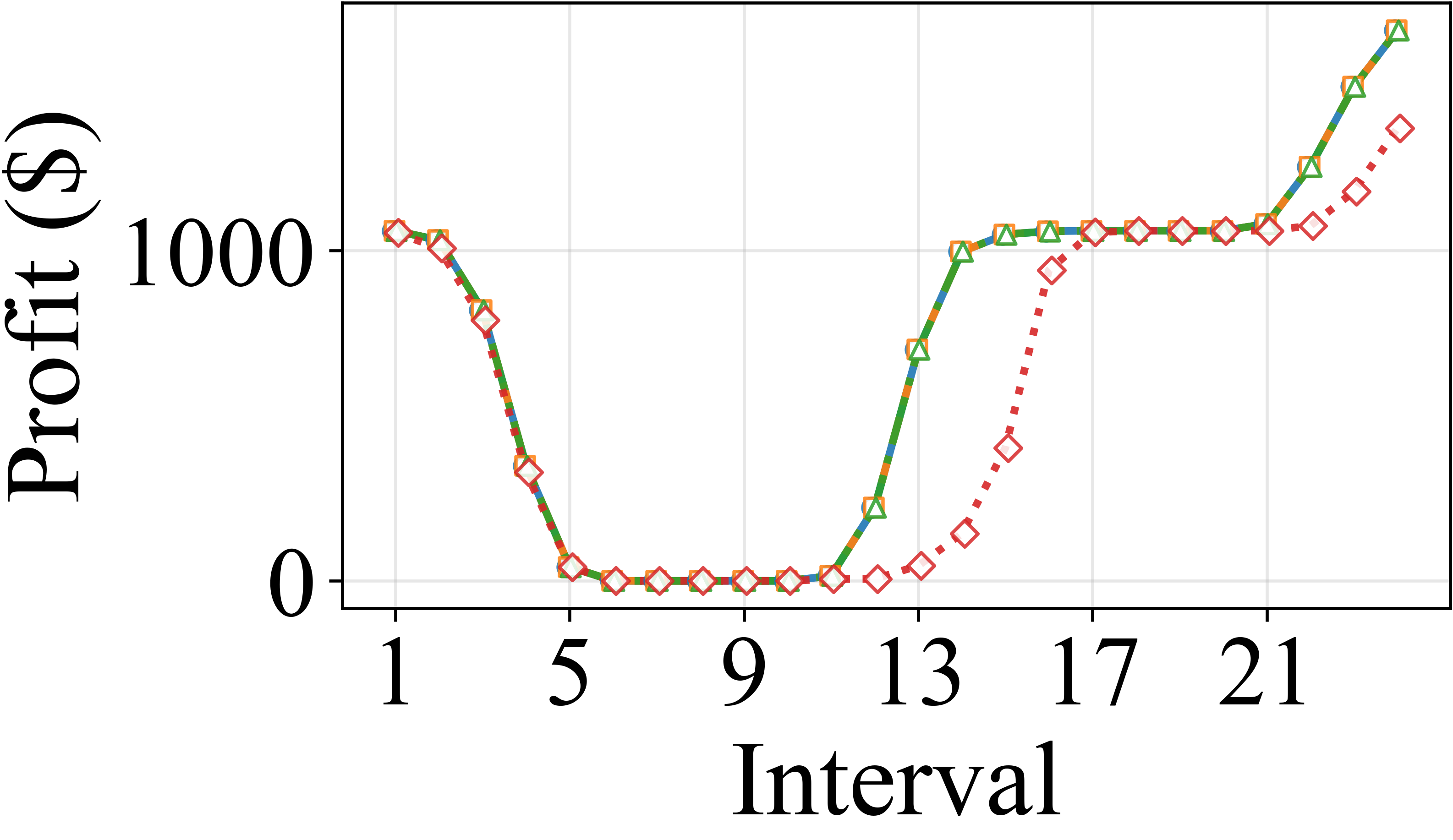}\hfill
  \includegraphics[width=0.48\columnwidth]{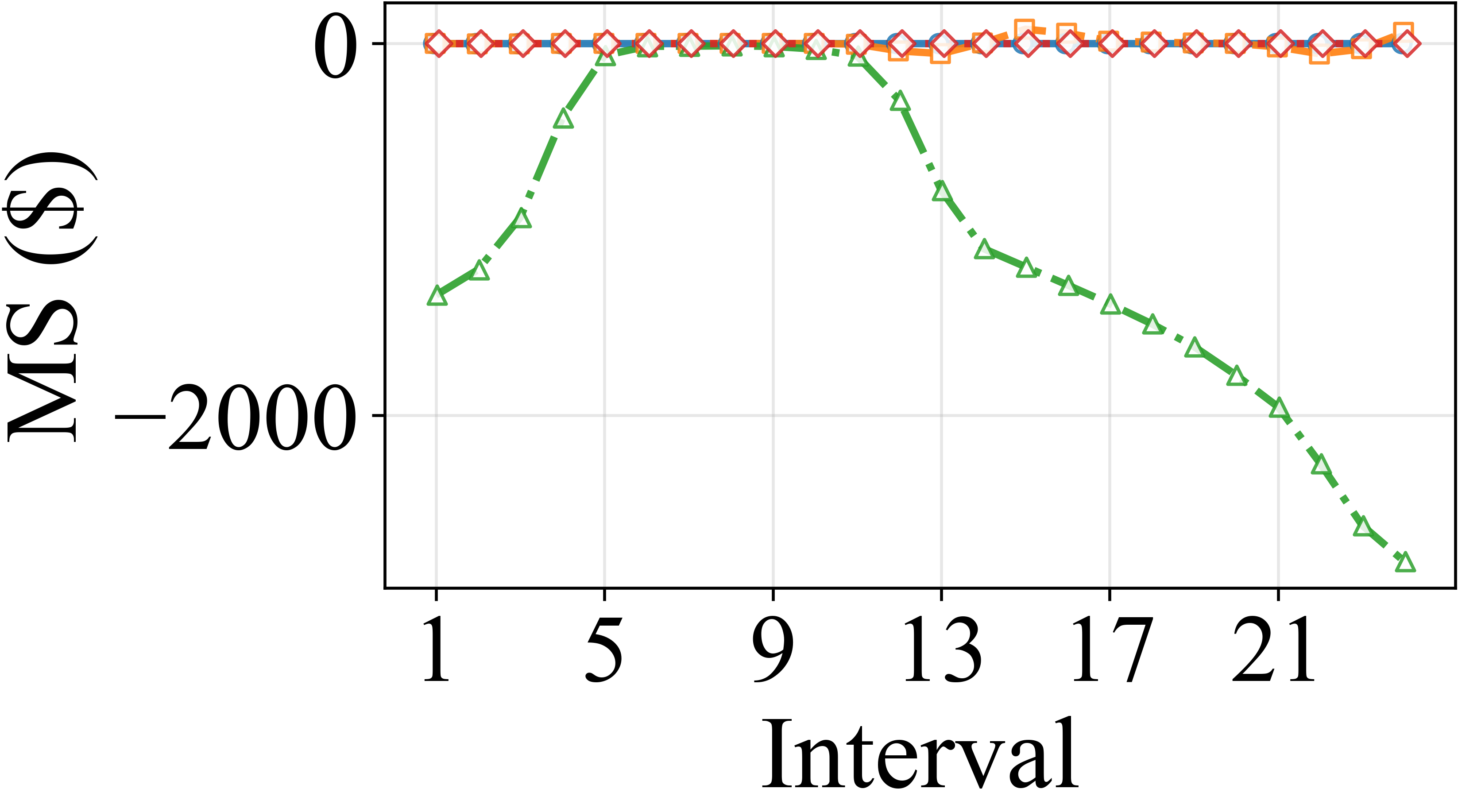}

  \vspace{0.15em}
  \includegraphics[width=0.72\columnwidth]{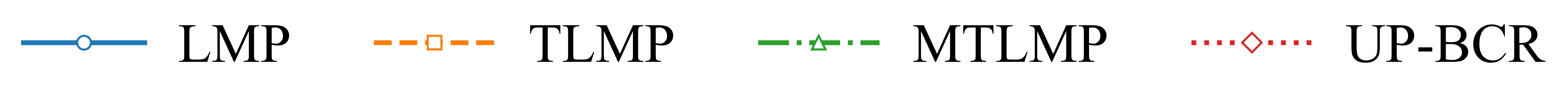}

  \caption{Top left: Mean total BCR per interval. {Top right: Distribution of total BCR across 200 scenarios. Bottom left: Mean generator profit per interval. Bottom right: Mean merchandising surplus per interval.}}
  \label{fig:sim}
  \vspace{-1.5em}
\end{figure}

% \begin{figure}[t]
%   \centering
%   \includegraphics[width=0.9\columnwidth]{figure/TotalBCR_mean.png}
%   \caption{Mean total BCR per interval.}
%   \label{fig:bcr_mean}
% \end{figure}

% \begin{figure}[t]
%   \centering
%   \includegraphics[width=0.9\columnwidth]{figure/Total_BCR_ESR_plus_Gen_total_boxplot.png}
%   \caption{Distribution of total BCR across 200 scenarios.}
%   \label{fig:bcr_box}
% \end{figure}

\begin{table}[t]
\centering
\footnotesize
\caption{Relationship between $T_{i,t'}$ and BCR for ESRs}
\label{tab:T_BCR}
\begin{tabular}{c|cc|cc}
\hline
 & \multicolumn{2}{c|}{ESR1} & \multicolumn{2}{c}{ESR2} \\
 & BCR$>0$ & BCR$=0$ & BCR$>0$ & BCR$=0$ \\
\hline
$T_{i,t'}=1$ & 100\% & 31\% & 100\% & 27\% \\
$T_{i,t'}=0$ & 0\%   & 69\% & 0\%   & 73\% \\
\hline
\end{tabular}

\vspace{0.5ex}
{\footnotesize Note: BCR is computed under LMP; UP-BCR yields the same ESR BCR categorization in these simulations. For each ESR, Table~\ref{tab:T_BCR} reports the percentage of intervals with $T_{i,t'}=1$ or $T_{i,t'}=0$.}
\vspace{-1.5em}
\end{table}

 Fig.~\ref{fig:sim} (top left) shows the mean total BCR per interval across 200 scenarios. Under LMP, positive BCR is concentrated in two groups of intervals, around 12--15 
and 21--24. In contrast, UP-BCR substantially reduces the BCR in 
these intervals, although it does not eliminate it completely. TLMP and MTLMP 
achieve zero BCR throughout by construction.

To explain why positive BCR occurs in these intervals, we relate Fig.~\ref{fig:sim} (top left) 
to the window-level indicator $T_{i,t'}$ in Section~\ref{subsec:binding_indicator} and 
its empirical relationship with BCR summarized in Table~\ref{tab:T_BCR}. Since $T_{i,t'}=1$ indicates an active SOC-boundary event within the rolling window, the relevant binding intertemporal constraints are the SOC constraints in \eqref{eq:RED_soc_dyn}--\eqref{eq:RED_soc_bounds}. 
Table~\ref{tab:T_BCR} shows that all observed cases with positive BCR occur when 
$T_{i,t'}=1$, whereas no positive-BCR case is observed when $T_{i,t'}=0$ for either ESR. 
Hence, in our simulations, positive BCR is observed only when these SOC constraints 
become binding.

Table~\ref{tab:T_BCR} also shows that $T_{i,t'}=1$ is not sufficient
for positive BCR, since BCR remains zero in a non-negligible fraction of intervals with
$T_{i,t'}=1$. This indicates that the binding of these SOC constraints alone does not
guarantee positive BCR.

Fig.~\ref{fig:sim} (top right) further shows the distribution of total BCR across the 200 
scenarios. LMP exhibits the highest median and the largest spread, whereas  UP-BCR significantly reduces total BCR and tightens the distribution. Following Corollary~\ref{cor:impossibility}, UP-BCR does not achieve zero BCR. TLMP and MTLMP are two nonuniform prices for ESRs that achieve zero BCR.
% ------------------------------------------------------------
\subsection{Price Trajectories}
% ------------------------------------------------------------

Fig.~\ref{fig:demand} (bottom) shows the mean price trajectories. LMP and the MTLMP discharge price $\bar\pi^{\mathrm{D}}_t$ nearly coincide, while the TLMP prices and the MTLMP charge price $\bar\pi^{\mathrm{C}}_t$ remain relatively low and stable. UP-BCR deviates more noticeably from LMP during the
two transition periods in which BCR is most pronounced.

% \begin{figure}[t]
%   \centering
%   \includegraphics[width=0.9\columnwidth]{figure/price_mean.png}
%   \caption{Mean price trajectories per interval.}
%   \label{fig:price}
% \end{figure}

% ------------------------------------------------------------
%\subsection{Generator profit, demand payment, and MS}
\subsection{Generator Revenue, Profit, Demand Payment, and MS}
% ------------------------------------------------------------

%\paragraph{Generator profit}

%Fig.~\ref{fig:sim} (bottom left) shows the mean generator profit across 200 scenarios. Generator revenue and profit exclude out-of-market BCR uplift payments. Relative to LMP, UP-BCR reduces generator profit during high-demand hours, whereas TLMP and MTLMP yield nearly the same generator profits as LMP. Generator revenue follows a similar pattern, while demand payments are lower under UP-BCR; additional results are provided in the appendix~\cite{YuLiuChen2026BCR}.
Fig.~\ref{fig:sim} (bottom left) shows the mean generator profit per interval across 200 scenarios. Generator revenue and profit include only in-market settlements and exclude out-of-market BCR uplift payments. Relative to LMP, UP-BCR reduces generator profit during high-demand hours, whereas TLMP and MTLMP yield nearly the same generator profits as LMP. Generator revenue follows a similar pattern, and demand payments are also lower under UP-BCR. Additional results are provided in the appendix.

%Relative to LMP, UP-BCR reduces generator profit during the high-demand hours. TLMP and MTLMP are nearly identical to LMP throughout. 
%The corresponding generator payment follows a similar pattern, and
%demand payment under the proposed uniform pricing rule is also lower
%than under LMP during the same periods. These additional simulation results are in the  appendix \cite{YuLiuChen2026BCR}.
 
% \begin{figure}[t]
%   \centering
%   \includegraphics[width=0.9\columnwidth]{figure/GenProfit_mean.png}
%   \caption{Mean generator profit per interval.}
%   \label{fig:gen_profit}
% \end{figure}
 
%\paragraph{Merchandising surplus}

Fig.~\ref{fig:sim} (bottom right) shows the merchandising
surplus (MS). LMP and UP-BCR maintain zero MS in the copper-plate model,
whereas MTLMP yields substantially negative MS during peak intervals.

% \begin{figure}[t]
%   \centering
%   \includegraphics[width=0.9\columnwidth]{figure/MS_mean.png}
%   \caption{Mean merchandising surplus per interval.}
%   \label{fig:ms}
% \end{figure}
\section{Conclusion}

This paper analyzes the fundamental tension between intertemporal operational efficiency and pricing transparency in rolling-window real-time electricity markets with ESRs. While rolling-window dispatch improves system efficiency by coordinating ESR decisions across time, it inherently creates intertemporal opportunity costs that cannot, in general, be fully internalized through a single uniform price. Our analysis shows that the existence of a zero-BCR uniform price depends on a restrictive alignment of marginal costs across charging and discharging decisions—conditions that are frequently violated in practice. As a result, out-of-market bid cost recovery (BCR) uplift payments are not merely an implementation artifact but a structural consequence of intertemporal coupling. To minimize negative impact brought by BCR, we analyze uniform pricing with minimized BCR (UP-BCR), which uses a transparent uniform pricing signal to reduce BCR uplifts in the real-time market.

%Uniform pricing with minimized BCR (UP-BCR) provides a practical compromise by minimizing BCR uplifts while preserving the simplicity and transparency of real-time market. More broadly, our findings suggest that the design of real-time markets must explicitly account for the growing role of intertemporal resources such as ESRs. In particular, we propose the intertemporal coupling indicator offering a useful lens for diagnosing when pricing inefficiencies and BCR uplifts are likely to arise.
We also propose an intertemporal
coupling indicator that provides a useful lens for diagnosing when pricing
inefficiencies and BCR uplifts are likely to arise. This study is limited to a copper-plate system, with numerical experiments based on a small number of resources and unity storage efficiencies. A complete theoretical characterization of the empirically observed relationship between the intertemporal coupling indicator and BCR, together with extensions to network-constrained systems and more realistic storage settings, is left for future work.

%These insights open future research for improved market designs aligning incentives under uncertainty and mitigating strategic behavior when costs are largely opportunity-based and privately held. As ESR penetration continues to increase, developing pricing frameworks that balance efficiency, transparency, and incentive compatibility will remain a central challenge for modern electricity markets.

%\begin{table}
%\caption{An Example of a Table}
%\label{table_example}
%\begin{center}
%\begin{tabular}{|c||c|}
%\hline
%One & Two\\
%\hline
%Three & Four\\
%\hline
%\end{tabular}
%\end{center}
%\end{table}

%%%%%%%%%%%%%%%%%%%%%%%%%%%%%%%%%%%%%%%%%%%%%%%%%%%%%%%%%%%%%%%%%%%%%%%%%%%%%%%%

\newpage
\appendix

\subsection{Linear Programming Reformulation of UP-BCR}
We show that minimizing the total BCR uplift defined in~\eqref{eq:bcr_def} is equivalent to the
linear program \eqref{eq:UP1}. Under uniform pricing
$\pi^\mathrm D_{i,t}=\pi^\mathrm C_{i,t}=\pi_t\in\mathbb{R}$, the total BCR uplift is
\begin{equation}
\begin{split}
F(\{\pi_t\})
:= \sum_{t=1}^{T}\Bigg(
&\sum_{i=1}^{N}\max\{0,-\Pi^{\mathrm{ESR}}_{i,t}(\pi_t)\}\\
&+\sum_{j=1}^{M}\max\{0,-\Pi^{\mathrm{G}}_{j,t}(\pi_t)\}
\Bigg),
\end{split}
\label{eq:up_bcr_nonsmooth}
\end{equation}
where $\Pi^{\mathrm{ESR}}_{i,t}(\pi_t)$ and $\Pi^{\mathrm{G}}_{j,t}(\pi_t)$
are affine in $\pi_t$. Since these profit functions are affine, the
constraints in~\eqref{eq:UP1} are linear, and hence \eqref{eq:UP1}
is a linear program.

\begin{lemma}\label{lem:lp_equiv}
Minimizing the total BCR uplift $F(\{\pi_t\})$ in
\eqref{eq:up_bcr_nonsmooth} over $\{\pi_t\}$ is equivalent to the
linear program \eqref{eq:UP1}: the two problems have the same optimal
objective value and the same set of optimal uniform prices
$\{\pi_t\}$.
\end{lemma}

\begin{proof}
Fix any $\{\pi_t\}$. By~\eqref{eq:UP1_esr}--\eqref{eq:UP1_nn},
\[
\begin{aligned}
\mathcal{B}^{\mathrm{ESR}}_{i,t}
&\ge \max\{0,-\Pi^{\mathrm{ESR}}_{i,t}(\pi_t)\},\\
\mathcal{B}^{\mathrm{G}}_{j,t}
&\ge \max\{0,-\Pi^{\mathrm{G}}_{j,t}(\pi_t)\}.
\end{aligned}
\]
Since \eqref{eq:UP1_obj} minimizes the sum of the BCR variables
$\mathcal{B}^{\mathrm{ESR}}_{i,t}$ and
$\mathcal{B}^{\mathrm{G}}_{j,t}$, each variable attains its smallest
feasible value at optimality:
\[
\begin{aligned}
\mathcal{B}^{\mathrm{ESR}*}_{i,t}
&= \max\{0,-\Pi^{\mathrm{ESR}}_{i,t}(\pi_t)\},\\
\mathcal{B}^{\mathrm{G}*}_{j,t}
&= \max\{0,-\Pi^{\mathrm{G}}_{j,t}(\pi_t)\}.
\end{aligned}
\]
Hence, for every fixed $\{\pi_t\}$,
\[
\min_{\substack{
\{\mathcal{B}^{\mathrm{ESR}}_{i,t}\},
\{\mathcal{B}^{\mathrm{G}}_{j,t}\}\\
\text{s.t.}~\eqref{eq:UP1_esr}\text{--}\eqref{eq:UP1_nn}}}
\sum_{t=1}^{T}
\left(
\sum_{i=1}^{N}\mathcal{B}^{\mathrm{ESR}}_{i,t}
+\sum_{j=1}^{M}\mathcal{B}^{\mathrm{G}}_{j,t}
\right)
=F(\{\pi_t\}).
\]
Minimizing both sides over $\{\pi_t\}$ shows that the UP-BCR problem
\eqref{eq:UP1} and the minimization of $F(\{\pi_t\})$ have the same
optimal objective value and the same set of optimal uniform prices.
\end{proof}

\subsection{Proof of Proposition~\ref{prop:zero_bcr}}

For notational convenience, let
$\mathcal B^{\mathrm D}_{i,t}:=\mathcal B^{\mathrm{ESR}}_{i,t}$
for $i\in\mathcal S_t^{\mathrm D,+}$ and
$\mathcal B^{\mathrm C}_{k,t}:=\mathcal B^{\mathrm{ESR}}_{k,t}$
for $k\in\mathcal S_t^{\mathrm C,+}$.

We prove both directions.

(\emph{Only if}) Suppose that a uniform price $\pi_t$ achieves
zero BCR for all dispatched participants, i.e., $\mathcal{B}_{j,t}^{\mathrm{G}} = 0$, $\forall j\in\mathcal{G}_t^\mathrm{+}$,
$\mathcal{B}_{i,t}^{\mathrm{D}} = 0$, $\forall i\in\mathcal{S}_t^\mathrm{D,+}$,
and $\mathcal{B}_{k,t}^{\mathrm{C}} = 0$, $\forall k\in\mathcal{S}_t^\mathrm{C,+}$. Here, $\mathcal{G}_t^\mathrm{+}$, $\mathcal{S}_t^\mathrm{D,+}$, and $\mathcal{S}_t^\mathrm{C,+}$ are the set of dispatched units defined in \eqref{eq:active_sets}. 

%(\emph{Only if}) Suppose that a uniform price $\pi_t$ achieves
%zero BCR for all active participants, where
%$\mathcal{G}_t^\mathrm{+}$, $\mathcal{S}_t^\mathrm{D,+}$, and
%$\mathcal{S}_t^\mathrm{C,+}$ are the active sets defined in
%\eqref{eq:active_sets}.

By the definition of BCR in \eqref{eq:bcr_def}, all dispatched participants must earn
nonnegative dispatch-following profit. Hence, we achieve the left-hand side inequalities below.
\begin{align}
  0 &\ge (c_j^{\mathrm{G}}-\pi_t)g_{j,t}^{\mathrm{RED}}, 
  &&\Rightarrow \pi_t \ge c_j^{\mathrm{G}}, 
  \quad \forall j\in \mathcal{G}_t^{+}, 
  \label{eq:pf_gen_app}\\
  0 &\ge (c_i^{\mathrm{D}}-\pi_t)g_{i,t}^{\mathrm{RED\text{-}D}}, 
  &&\Rightarrow \pi_t \ge c_i^{\mathrm{D}},
  \quad \forall i\in \mathcal{S}_t^\mathrm {D,+}, 
  \label{eq:pf_dis_app}\\
  0 &\ge (\pi_t-c_k^{\mathrm{C}})g_{k,t}^{\mathrm{RED\text{-}C}}, 
  &&\Rightarrow \pi_t \le c_k^{\mathrm{C}},
  \quad \forall k\in \mathcal{S}_t^\mathrm {C,+}. 
  \label{eq:pf_chg_app}
\end{align}
The right-hand side inequalities above are achieved because all realized dispatch quantities are strictly positive by definition of the set of dispatched units in \eqref{eq:active_sets}. Dividing dispatch quantities on both sides of the inequalities gives the right-hand side.
% through yields
% \begin{align}
%   \pi_t &\ge c_j^{\mathrm{G}}, && \forall j\in\mathcal{G}_t^\mathrm{+}, \label{eq:lb_gen_app}\\
%   \pi_t &\ge c_i^{\mathrm{D}}, && \forall i\in\mathcal{S}_t^\mathrm{D,+}, \label{eq:lb_dis_app}\\
%   \pi_t &\le c_k^{\mathrm{C}}, && \forall k\in\mathcal{S}_t^\mathrm{C,+}. \label{eq:ub_chg_app}
% \end{align}\

Therefore, the uniform price $\pi_t$ has
\[
\pi_t \ge \max\!\left\{
\max_{j\in\mathcal{G}_t^\mathrm{+}} c_j^{\mathrm{G}},\;
\max_{i\in\mathcal{S}_t^\mathrm{D,+}} c_i^{\mathrm{D}}
\right\},
\qquad
\pi_t \le \min_{k\in\mathcal{S}_t^\mathrm{C,+}} c_k^{\mathrm{C}}.
\]
Such a uniform price $\pi_t$ with zero BCR for all units can exist only if
\[
\max\!\left\{
\max_{j\in\mathcal{G}_t^\mathrm{+}} c_j^{\mathrm{G}},\;
\max_{i\in\mathcal{S}_t^\mathrm{D,+}} c_i^{\mathrm{D}}
\right\}
\le
\min_{k\in\mathcal{S}_t^\mathrm{C,+}} c_k^{\mathrm{C}}.
\]

(\emph{If}) Conversely, suppose that condition~\eqref{eq:feasibility}
holds. Then any price $\pi_t$ chosen in the interval
\[
\max\!\left\{
\max_{j\in\mathcal{G}_t^\mathrm{+}} c_j^{\mathrm{G}},\;
\max_{i\in\mathcal{S}_t^\mathrm{D,+}} c_i^{\mathrm{D}}
\right\}
\le \pi_t \le
\min_{k\in\mathcal{S}_t^\mathrm{C,+}} c_k^{\mathrm{C}}
\]
satisfies the right-hand side for \eqref{eq:pf_gen_app}--\eqref{eq:pf_chg_app}. Therefore,
every dispatched generator and discharging ESR receives a price no
smaller than its bid cost, while every dispatched charging ESR pays
a price no larger than its bid cost. Hence all dispatched participants
have nonnegative dispatch-following profit, so zero BCR is achieved by the uniform price $\pi_t$.
\hfill$\square$

\subsection{Proof of Corollary~\ref{cor:impossibility}}

%For notational convenience, let
%$\mathcal B^{\mathrm D}_{i,t}:=\mathcal B^{\mathrm{ESR}}_{i,t}$
%for $i\in\mathcal S_t^{\mathrm D,+}$ and
%$\mathcal B^{\mathrm C}_{k,t}:=\mathcal B^{\mathrm{ESR}}_{k,t}$
%for $k\in\mathcal S_t^{\mathrm C,+}$.

Suppose condition~\eqref{eq:feasibility} fails in the 
nontrivial case described in Remark~\ref{rem:nontrivial_case}, for which 
Section~\ref{sec:toy} provides a concrete example---a set of dispatched units configuration with one charging ESR and one discharging ESR for which condition~\eqref{eq:feasibility} fails. 
By Proposition~\ref{prop:zero_bcr}, there does not exist any uniform price $\pi_t$ 
that achieves zero BCR for all dispatched participants at interval~$t$.

Therefore, the point
\[
\sum_{j\in\mathcal{G}_t^\mathrm{+}} \mathcal{B}_{j,t}^{\mathrm{G}}
+\sum_{i\in\mathcal{S}_t^\mathrm{D,+}} \mathcal{B}_{i,t}^{\mathrm{D}}
+\sum_{k\in\mathcal{S}_t^\mathrm{C,+}} \mathcal{B}_{k,t}^{\mathrm{C}}
=0
\]
is infeasible for the UP-BCR problem.

Since all BCR variables are constrained to be nonnegative, the objective value of the UP-BCR problem is also nonnegative. If the optimal objective value were equal to zero, then every BCR term would have to be zero, which would imply the existence of a uniform price achieving zero BCR for all dispatched participants, contradicting Proposition~\ref{prop:zero_bcr}. Hence, the optimal objective value must satisfy
\[
\min_{\pi_t}
\left(
\sum_{j\in\mathcal{G}_t^\mathrm{+}} \mathcal{B}_{j,t}^{\mathrm{G}}
+\sum_{i\in\mathcal{S}_t^\mathrm{D,+}} \mathcal{B}_{i,t}^{\mathrm{D}}
+\sum_{k\in\mathcal{S}_t^\mathrm{C,+}} \mathcal{B}_{k,t}^{\mathrm{C}}
\right)
>0.
\]
\hfill$\square$

\subsection{Additional Simulation Results}

Fig.~\ref{fig:gen_pay} shows the mean generator revenue per
interval across 200 scenarios. Its pattern is similar to that of
generator profit in the main text.
\begin{figure}[t]
  \centering
  \includegraphics[width=0.8\columnwidth]{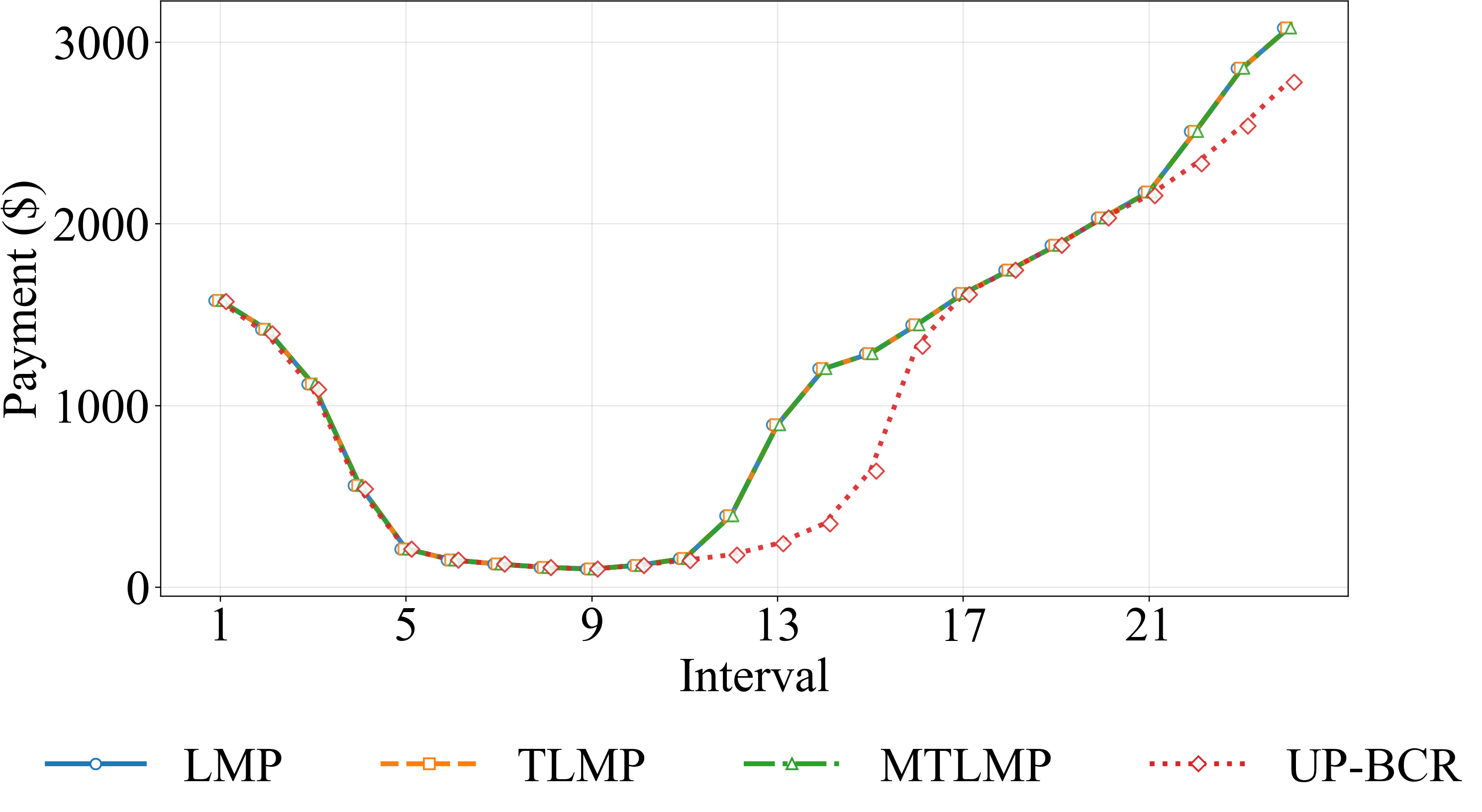}
  \caption{Mean generator revenue per interval.}
  \label{fig:gen_pay}
\end{figure}

Fig.~\ref{fig:demand_pay} shows the mean demand payment.
\begin{figure}[t]
  \centering
  \includegraphics[width=0.8\columnwidth]{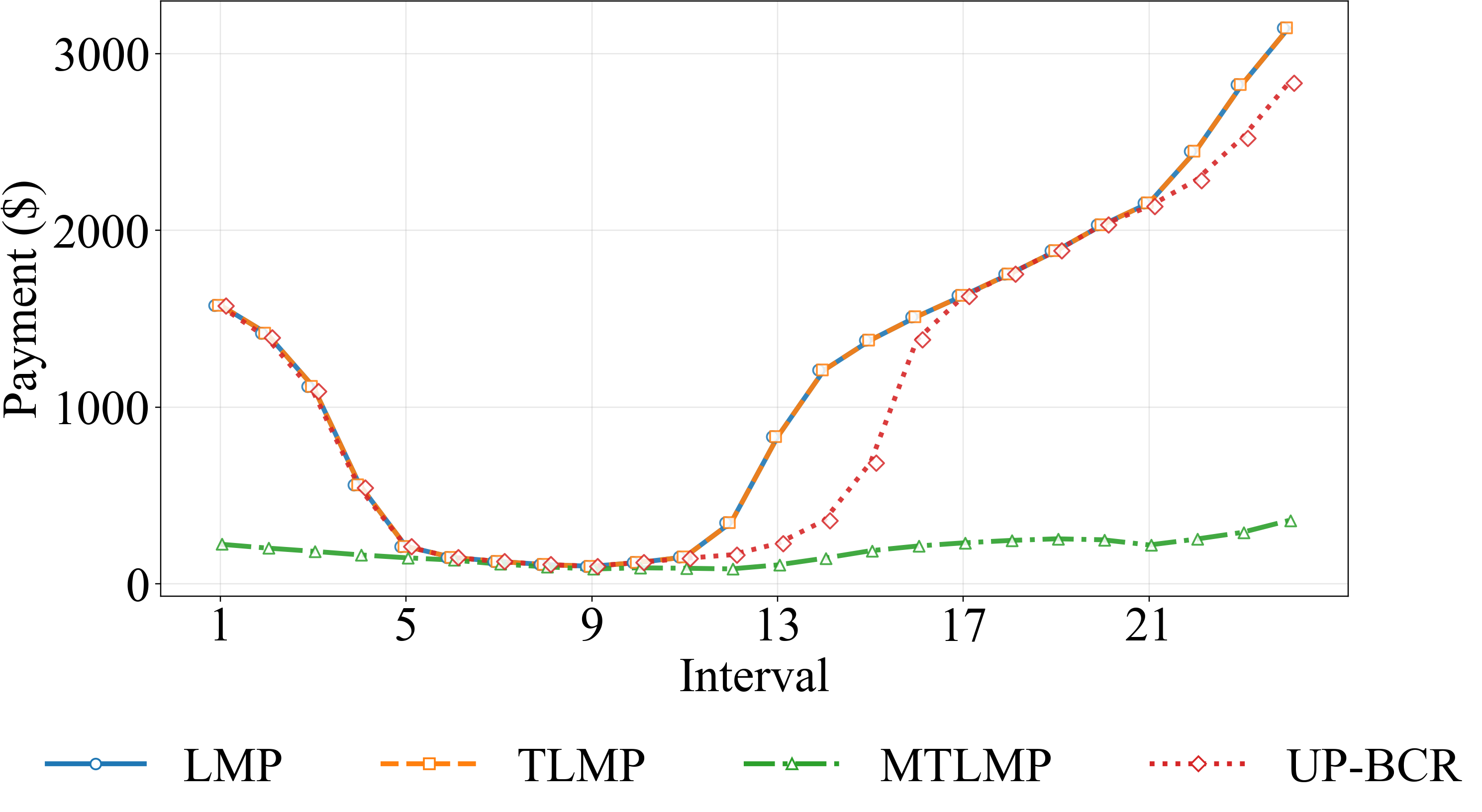}
  \caption{Mean demand payment per interval.}
  \label{fig:demand_pay}
\end{figure}
\section*{Acknowledgment}
The authors gratefully acknowledge Sergio Dueñas Meléndez, Storage Sector Manager at CAISO, for valuable discussions on practical energy policy considerations and risks in new market mechanism designs for integrating energy storage into electricity markets.


\begin{thebibliography}{99}
\bibitem{EIA2025Battery}
U.S. Energy Information Administration, ``U.S. battery capacity increased 66\% in 2024,'' {\it Today in Energy}, Mar. 12, 2025. [Online]. Available: \url{https://www.eia.gov/todayinenergy/detail.php?id=64705}

\bibitem{CAISO2025BatteryReport}
California Independent System Operator (CAISO), ``2024 special report on battery storage,'' May 29, 2025. [Online]. Available: \url{https://www.caiso.com/documents/2024-special-report-on-battery-storage-may-29-2025.pdf}

\bibitem{Hogan:20}
W.W. Hogan, ``Electricity Market Design: Multi-Interval Pricing Models,'' Online paper, Harvard Electricity Policy Group, Jun. 2020. [Online]. Available: \url{https://scholar.harvard.edu/files/whogan/files/hogan_hepg_multi_period_062220.pdf}

\bibitem{Zhao2019multi}
J. Zhao, T. Zheng, and E. Litvinov, ``A Multi-Period Market Design for Markets With Intertemporal Constraints,'' {\it IEEE Trans. Power Syst.}, vol. 35, no. 4, pp. 3015--3025, Jul. 2020.





\bibitem{CAISO2025StorageDesign}
California Independent System Operator (CAISO), ``Storage design and modeling: Working group on uplift \& default energy bids (DEB), state-of-charge management, and mixed-fuel \& distribution-level resources,'' 2025. [Online]. Available: \url{https://stakeholdercenter.caiso.com/StakeholderInitiatives/Storage-design-modeling}

\bibitem{CAISO2024StorageBCR}
California Independent System Operator (CAISO), ``Storage Bid Cost Recovery and Default Energy Bids Enhancements,'' Nov. 2024. [Online]. Available: \url{https://stakeholdercenter.caiso.com/StakeholderInitiatives/storage-bid-cost-recovery-and-default-energy-bids-enhancements}

\bibitem{Guo&Chen&Tong:21TPS}
Y. Guo, C. Chen, and L. Tong, ``Pricing Multi-Interval Dispatch Under Uncertainty Part I: Dispatch-Following Incentives,'' {\it IEEE Trans. Power Syst.}, vol. 36, no. 5, pp. 3865--3877, Sep. 2021.

% \bibitem{Chen&Guo&Tong:20TPS}
% C. Chen, Y. Guo, and L. Tong, ``Pricing Multi-Interval Dispatch Under Uncertainty Part II: Generalization and Performance,'' {\it IEEE Trans. Power Syst.}, vol. 36, no. 5, pp. 3878--3886, Sep. 2021.

\bibitem{PricingESR}
C. Chen, and L. Tong, ``Pricing Real-Time Stochastic Storage Operations,'' {\it in Electric Power Systems Research}, vol. 212, pp. 108606, 2022. 

\bibitem{Werner2023}
L. Werner, N. Christianson, A. Zocca, A. Wierman, and S. Low, ``Pricing uncertainty in stochastic multi-stage electricity markets,'' {\it in Proc. 62nd IEEE Conference on Decision and Control (CDC)}, 2023, pp. 1580--1587.

\bibitem{CMP}
B. Hua, D.A. Schiro, T. Zheng, R. Baldick, and E. Litvinov, ``Pricing in Multi-Interval Real-Time Markets,'' {\it IEEE Trans. Power Syst.}, vol. 34, no. 4, pp. 2696--2705, Jul. 2019.

\bibitem{ChoPapavasiliou2022}
J. Cho and A. Papavasiliou, ``Pricing Under Uncertainty in Multi-Interval Real-Time Markets,'' {\it Oper. Res.}, vol. 71, no. 6, Nov.--Dec. 2023.

\bibitem{Chen26ramping}
C. Chen, V. Norambuena, and L. Tong, ``Ramping procurement and bid-cost recovery in real-time market,'' arXiv:2606.19599, 2026.

%\bibitem{Chen26ramping}
%C. Chen, V. Norambuena, and L. Tong, ``Ramping procurement and bid cost recovery in real-time dispatch,'' arXiv preprint, 2026.
% \bibitem{ChenTong2025IncentivizingRamping}
% C. Chen and L. Tong, ``Incentivizing Ramping with Uniform Pricing,'' {\it in Proc. IEEE Power \& Energy Society General Meeting (PESGM)}, 2025, pp. 1--5.

% \bibitem{Wasti2025}
% S. Wasti, A. R. Kumar, S. Varghese, A. Giacomoni, and M. Webster, ``Energy storage state-of-charge management in real-time markets,'' {\it in Proc. IEEE Power \& Energy Society General Meeting (PESGM)}, 2025, pp. 1--5.


% \bibitem{ChenTong2023SOCBid}
% C. Chen and L. Tong, ``Convexifying market clearing of SoC-dependent bids from merchant storage participants,'' {\it IEEE Trans. Power Syst.}, vol. 38, no. 3, pp. 2955--2957, 2023.
%\bibitem{YuLiuChen2026BCR}
%\textcolor{cyan}{Y. Yu, J. Liu, and C. Chen, “Appendix: Minimizing Bid Cost Recovery for Energy Storage with Uniform Pricing,” arXiv:2606.19599, 2026.}



\end{thebibliography}
\end{document}